%% file: submission-arxiv.tex
\documentclass[12pt]{article}
\usepackage[english]{babel}
\usepackage[utf8]{inputenc}
\usepackage{amssymb,amsmath,amsthm}
\usepackage{thmtools}
\usepackage{thm-restate}
\theoremstyle{plain}
\newtheorem{theorem}{Theorem}[section]
\newtheorem{lemma}[theorem]{Lemma}
\newtheorem{corollary}[theorem]{Corollary}
\newtheorem{claim}[theorem]{Claim}
\newtheorem{fact}[theorem]{Fact}

\theoremstyle{definition}
\newtheorem{definition}[theorem]{Definition}
\usepackage{mathtools} 
\usepackage{bbm}
\usepackage[backref,colorlinks,citecolor=blue,bookmarks=true,linktocpage]{hyperref}
\usepackage{url}
\usepackage[algo2e,ruled]{algorithm2e} 
\usepackage{todonotes}
\usepackage{natbib}
\usepackage{fullpage}
\usepackage{mleftright}
\usepackage{microtype}
\usepackage{verbatim}
\newcommand{\Paren}[1]{\mleft(#1\mright)}
\newcommand{\bigO}[1]{O\Paren{#1}}

\newcommand{\eps}{\varepsilon}

\newcommand{\norm}[1]{\mleft\|#1\mright\|}
\newcommand{\ignore}[1]{}

\newcommand{\eqdef}{\coloneqq}
\renewcommand{\leq}{\leqslant}
\renewcommand{\geq}{\geqslant}

\let\orgvec\vec

\allowdisplaybreaks

\newcommand{\new}[1]{#1}

\title{Independence Testing for Bounded Degree Bayesian Network}
\author{Arnab Bhattacharyya
\and Cl\'ement L. Canonne
\and Joy Qiping Yang}
\usepackage{times}

\newcommand{\assign}{:=}
\newcommand{\cdummy}{\cdot}
\newcommand{\infixand}{\text{ and }}

\newcommand{\nin}{\not\in}

\newcommand{\tmcolor}[2]{{\color{#1}{#2}}}

\newcommand{\tmem}[1]{{\em #1\/}}

\newcommand{\tmop}[1]{\ensuremath{\operatorname{#1}}}

\newcommand{\p}{\mathbf{p}}

\newcommand{\of}{:}
\newcommand{\E}{\mathbb{E}}

\renewcommand{\todo}[1]{\marginpar{\tiny\tmcolor{orange}{TODO:} #1}}

\begin{document}

\maketitle

\begin{abstract}
We study the following {\em independence testing} problem: given access to samples from a distribution $P$ over $\{0,1\}^n$, decide whether $P$ is a product distribution or whether it is $\eps$-far in total variation distance from any product distribution. For arbitrary distributions, this problem requires $\exp(n)$ samples. We show in this work that if $P$ has a sparse structure, then in fact only linearly many samples are required.
Specifically, if $P$ is Markov with respect to a Bayesian network whose underlying DAG has in-degree bounded by $d$, then $\tilde{\Theta}(2^{d/2}\cdot n/\eps^2)$ samples are necessary and sufficient for independence testing.
\end{abstract}

\section{Introduction}
\input{sec-intro}

\section{Preliminaries}
\input{sec-preliminaries}

\section{Upper Bound for Independence Testing}
\input{sec-upperbound}

\section{The $\Omega(2^{d/2}{n} / \eps^2)$ Lower Bound}
    \label{sec:lb:outline}
\input{sec-lowerbound}

%
%
%
%
%
%

\section{Useful results on the MGFs of Binomials and Multinomials}
    \label{ssec:tools:mgf}
\input{sec-mgfs}

\section*{Acknowledgments}
Yang would like to thank Vipul Arora and Philips George John for the helpful discussions on the lower bound analysis; and Vipul, specifically, for providing valuable feedback on the manuscript.

\bibliographystyle{plainnat}
\bibliography{references}

\appendix

\section{Deferred Proofs}
\input{app-appendix_deferred_proofs}
\subsection{An $\Omega(2^{d/2}\sqrt{n} / \eps^2)$ Lower Bound} \label{section:lb_2}
\input{sec-lower_bound_warmup}
%
%
%

\end{document}

%% file: sec-intro.tex
\ignore{

      \begin{itemize}
      \item
      Fundamental importance of independence testing. Core component in algorithms for variable selection and causal discovery.
      \item
       	Define the problem formally.
	\item
	    Claim that without any distributional assumptions, exponentially many samples needed. Point out that this is common in statistics and algorithms, and often improvements possible if distribution has sparse structure.
	    \item
	    Introduce Bayesian networks. Mention they succinctly describe how samples are generated. Explain that we restrict independence testing to distributions consistent with a bounded-degree Bayes net.       	
		\item
		Give formal theorem
		\item
		Overview of upper bound. 
		\item 
		Overview of lower bound      
      \end{itemize}
 }

It is often convenient to model high-dimensional datasets as probability distributions. An important reason is that the language of probability formalizes what we intuitively mean by two features of the data being {\em independent}: the marginal distributions along the two corresponding coordinates are statistically independent. Independence is a basic probabilistic property that, when true, enables better interpretability of data as well as computationally fast inference algorithms. 

In this work, we study the problem of testing whether a collection of $n$ binary random variables are mutually independent. Independence testing is an old and foundational problem in hypothesis testing \citep{neyman1933ix, lehmann2005testing}, and we consider it from the perspective of {\em property testing} \citep{rubinfeld1996robust, goldreich1998property}. Our goal is to design a tester that accepts product distributions on $n$ variables and rejects distributions that have statistical distance at least $\eps$ from every product distribution. The tester gets access to i.i.d.~samples from the input distribution and can fail (in either case) with probability at most $1/3$. The objective is to minimize the number of samples as a function of $n$ and $\eps$. 

Previous work \citep{diakonikolas2016new} shows\footnote{See Appendix \ref{sec:gen} for a more general result.} that testing independence of a distribution on $\{0,1\}^n$ requires $2^{\Omega(n)}$ queries which is intractable. However, there is an aspect of this proof that is unsatisfying. The ``hard distributions'', which are shown to require many queries to distinguish from product distributions, cannot be described succinctly, as they require exponentially many bits to describe. Thus, a natural question arises: can we test independence efficiently for distributions on $\{0,1\}^n$ that have a sparse description?

One of the most canonical ways to describe high-dimensional distributions is as {\em Bayesian networks} (or {\em Bayes nets} in short). A Bayes net specifies how to generate an $n$-dimensional sample in an iterative way and is especially useful for modeling causal relationships. Formally, a Bayes net on $\{0,1\}^n$ is given by a directed acyclic graph (DAG) $G$ on $n$ vertices and probability distributions $p_{i,\pi}$ on $\{0,1\}$ for all $i \in [n]$ and all assignments $\pi$ to the parents of the $i$'th node in the graph $G$. An $n$-dimensional sample is obtained by sampling the nodes in a topological order of $G$, where the $i$'th node is sampled according to $p_{i,\pi}$ for the assignment $\pi$ that is already fixed by the samples of the parent nodes of $i$. The generated distribution on $\{0,1\}^n$ is said to be {\em Markov with respect to} $G$.

In this work, we consider independence testing on distributions having a sparse Bayes net description, a class of distributions naturally arising in, and with numerous applications to, machine learning~\citep{WainwrightJ08}, robotics, natural language processing, medicine, and  biological settings such as gene expression data~\citep{genexpression1,genexpression2,geneexpression3} (where it is known that most genetic networks are actually sparse). In these cases, one hopes to leverage that additional knowledge to test whether those sparse, local dependencies are actually present without having to pay the prohibitive exponential cost in the dimension $n$. Specifically, we analyze independence testing on distributions that are promised to be Markov with respect to a DAG of maximum in-degree $d$, where $d \ll n$. While the learning sample complexity, known to be $\tilde{O}(2^d \cdot n / \eps^2)$~\citep{BhattacharyyaGMV20}, provides a baseline for the \emph{testing} question, it is not at all obvious that this is tight, and what the correct dependence on $d$ and even $n$ are. Our main result essentially settles this question, and establishes the following: 
\begin{theorem}[Informal Main Theorem]
    \label{theo:main:informal}
Suppose an unknown distribution $P$ on $\{0,1\}^n$ is Markov with respect to an unspecified degree-$d$ DAG.
 The sample complexity of testing whether $P$ is a product
      distribution or is at least $\eps$-far from any product distribution is $\tilde{\Theta} (2^{d / 2} \cdot n /
      \eps^2)$.
\end{theorem}
In the course of proving this theorem, we additionally derive several technical results that are of independent interest, such as bounds on the moment generating function of squared binomials and an independence testing algorithm for arbitrary distributions in Hellinger distance.

Our work explicitly initiates the study of {\em testing graphical structure} in the context of property testing. That is, instead of testing a statistical property of a distribution (as is typical in distribution property testing), we can interpret our problem as that of testing a graphical property of the underlying graph that describes the distribution. \new{In particular, testing independence can be viewed as testing maximum degree-0 of a distribution's graph.} This point of view opens the door to testing many other relevant graphical properties of graphical models, e.g., maximum degree $k$, being a forest, being connected, etc. \new{Hence, by analyzing the ``base case'' of independence testing, we provide a necessary first step to solving other graphical testing problems.} Information-theoretic bounds for related problems were studied recently by \cite{neykov2019combinatorial}.

\subsection{Related Work}
Distribution testing has been an active and rapidly progressing research program for the last 20+ years; see \cite{rubinfeld2012taming} and \cite{canonne2020survey} for surveys. One of the earliest works in this history was that of \cite{batu2001testing} who studied testing independence of two random variables. There followed a series of papers \citep{alon2007testing, levi2013testing, AcharyaDK15}, strengthening and generalizing bounds for this problem, culminating in the work of \cite{diakonikolas2016new} who gave tight bounds for testing independence of distributions over $[n_1]\times \cdots \times [n_d]$. \citet{HaoL20} recently considered the (harder) problem of estimating the distance to the closest product distribution (i.e., \emph{tolerant} testing), showing this task could, too, be performed with a sublinear sample complexity.

Though most of the focus has been on testing properties of arbitrary input distributions, it has long been recognized that distributional restrictions are needed to obtain sample complexity improvements. For example, \cite{rubinfeld2009testing, adamaszek2010testing} studied testing uniformity of monotone distributions on the hypercube. Similarly, \cite{daskalakis2012learning} considered the problem of testing monotonicity of $k$-modal distributions. The question of independence testing of structured high-dimensional distributions was considered in~\citet{DaskalakisD019} in the context of Ising models. We note that while their work is in the same spirit as ours, Ising models and Bayes nets are incomparable modeling assumptions, and their results (and techniques) and ours \new{are mostly disjoint. Further, while the results may overlap in some special cases, the conversion between parameterizations makes them difficult to compare even in these cases (e.g., dependence on the maximum edge value parameter $\beta$ for Ising models, and max-undirected-degree vs. max-in-degree).} More recently, \cite{CanonneDKS20, daskalakis2016square, BhattacharyyaGMV20, bhattacharyya2021testing} have studied identity testing and closeness testing for distributions that are structured as degree-$d$ Bayes nets. Our work here continues this research direction in the context of independence testing.

\subsection{Our techniques}
\noindent\textit{Upper bound.} The starting point of our upper bound is the (standard) observation that a distribution $P$ over $\{0,1\}^n$ is far from being a product if, and only if, it is far from the product of its marginals. By itself, this would not lead to any savings over the trivial exponential sample complexity. However, we can combine this with a localization result due to~\cite{daskalakis2016square}, which then guarantees that if the degree-$d$ Bayes net $P$ is at \emph{Hellinger} distance $\eps$ from the product of its marginals $P'$, then there exists \emph{some} vertex $i\in[n]$ such that $P_{i,\Pi_i}$ (the marginalization of $P$ onto the set of nodes consisting of $i$ and its $d$ parents) is at Hellinger distance at least $\frac{\eps}{\sqrt{n}}$ from $P'_{i,\Pi_i}$. 
These two facts, combined, seem to provide exactly what is needed: indeed, given access to samples from $P$ and any fixed set of $d+1$ vertices $S$, one can simulate easily samples from both $P_S$ and $P'_S$ (for the second, using $d+1$ samples from $P$ to generate one from $P'_S$, as $P'_S$ is the product of marginals of $P_S$). A natural idea is then to iterate over all $\binom{n}{d+1}$ possible subsets $S$ of $d+1$ variables and check whether $P_S=P'_S$ for each of them using a closeness testing algorithm for arbitrary distributions over $\{0,1\}^{d+1}$: the overhead due to a union bound and the sampling process for $P_S,P'_S$ adds a factor $O(d\cdot \log\binom{n}{d+1}) = O(d^2\log n)$ to the closeness testing procedure. However, since testing closeness over $\{0,1\}^{d+1}$ to total variation distance $\eps'$ has sample complexity $O(2^{2d/3}/\eps'^2)$ and, by the quadratic relation between Hellinger and total variation distances, we need to take $\eps' = \frac{\eps^2}{n}$, we would then expect the overall test to result in a $\tilde{O}(2^{2d/3}n^2/\eps^4)$ sample complexity~--~much more than what we set out for.

\new{A first natural idea to improve upon this is to use the refined
identity testing result of~\citet[Theorem 4.2]{daskalakis2016square} for Bayes nets,
which avoids the back and forth between Hellinger and total variation
distance and thus saves on this quadratic blowup. Doing so, we could in the last step keep $\eps' =\frac{\eps}{\sqrt{n}}$ (saving on this quadratic blowup), and pay only overall $\tilde{O}(2^{3d/4}/\eps'^2)=\tilde{O}(2^{3d/4}n/\eps^2)$. This is better, but still falls short of our original goal.}

The second idea is to forego closeness testing in the last step entirely, and instead use directly an \emph{independence} testing algorithm for arbitrary distributions over $\{0,1\}^{d+1}$, to test if $P_S$ is indeed a product distribution for every choice of $S$ considered. Unfortunately, while promising, this idea suffers from a similar drawback as our very first attempt: namely, the known independence testing algorithms are all designed for testing in total variation distance (not Hellinger)! Thus, even using an optimal TV testing algorithm for independence~\citep{AcharyaDK15,diakonikolas2016new} would still lead to this quadratic loss in the distance parameter $\eps'$, and a resulting $\tilde{O}(2^{d/2}n^2/\eps^4)$ sample complexity.

To combine the best of our last two approaches and achieve the claimed $\tilde{O}(2^{d/2}n/\eps^2)$, we combine the two insights and perform, for each set $S$ of $d+1$ variables, independence testing on $P_S$ in Hellinger distance. In order to do so, however, we first need to design a testing algorithm for this task, as none was previously available in the literature. Fortunately for us, we are able to design such a testing algorithm (Lemma~\ref{lemma:hellinger_independence_tester}) achieving the desired~--~and optimal~--~sample complexity. Combining this Hellinger independence testing algorithm over $\{0,1\}^{d+1}$ with the above outline finally leads to the $\tilde{O}(2^{d/2}n/\eps^2)$ upper bound of Theorem~\ref{theo:main:informal}.\medskip

\noindent\textit{Lower bound.} To obtain our $\Omega(2^{d/2}n/\eps^2)$ lower bound on testing independence of a degree-$d$ Bayes net, we start with the construction introduced by \cite{CanonneDKS20} to prove an $\Omega(n/\eps^2)$ sample complexity lower bound on testing \emph{uniformity} of degree-1 Bayes nets. At a high level, this construction relies on picking uniformly at random a perfect matching $M$ of the $n$ vertices, which defines the structure of the Bayes net; and, for each of the $n/2$ resulting edges, picking either a positive or negative correlation (with value $\pm \eps/\sqrt{n}$) between the two vertices, again uniformly at random. One can check relatively easily that every Bayes net $P_{\lambda}$ obtained this way, where $\lambda$ encodes the matching $M$ and the $2^{n/2}$ signs, is (1)~a degree-$1$ Bayes net, (2)~at total variation $\eps$ from the uniform distribution $U$. The bulk of their analysis then lies in showing that (3)~$\Omega(n/\eps^2)$ samples are necessary to distinguish between $U$ and such a randomly chosen $P_\lambda$. Generalizing this lower bound construction and analysis to independence testing (not just uniformity), and to degree-$d$ (and not just degree-$1$) Bayes nets turns out to be highly non-trivial, and is our main technical contribution.

Indeed, in view of the simpler and different $\Omega(2^{d/2}\sqrt{n}/\eps^2)$ sample complexity lower bound for uniformity testing degree-$d$ Bayes nets obtained in~\cite{CanonneDKS20} \emph{when the structure of the Bayes net is known},\footnote{We note that generalizing this (weaker, in view of the dependence on $n$) $\Omega(2^{d/2} \sqrt{n}/\eps^2)$ sample complexity lower bound to our setting is relatively simple, and we do so in Appendix~\ref{section:lb_2} (specifically, Theorem~\ref{theorem:lb_2}).} one is tempted to adapt the same idea to the matching construction: that is, reserve $d-1$ out of the $n$ vertices to ``encode'' a pointer towards one of $2^{d-1}$ independently chosen $P_\lambda$'s as above (i.e., the hard instances are now uniform mixtures over $2^{d-1}$ independently generated degree-$1$ hard instances).
The degree-$1$ hard instances lead in previous work to a tighter dependence on $n$ (linear instead of $\sqrt{n}$) because their Bayesian structure is \emph{unknown}.
Thus, by looking at the mixture of these degree-$1$ Bayes nets, one could hope to extend the analysis of that second lower bound from~\cite{CanonneDKS20} and get the desired  $\Omega(2^{d/2}n/\eps^2)$ lower bound.
Unfortunately, there is a major issue with this idea: namely, if the $2^{d-1}$ matchings are chosen independently, then the resulting overall distribution is unlikely to be a degree-$d$ Bayes net~--~instead, each vertex will have expected degree $\Omega(2^d)$! (This was not an issue in the corresponding lower bound of~\cite{CanonneDKS20}, since for them each of the $2^d$ components of the mixture was a degree-$0$ Bayes net, i.e., a product distribution; so degrees could not ``add up'' across the components).

To circumvent this, we instead choose the matching $M$ to be common to all $2^{d-1}$ components of the mixture and only pick the sign of their $n/2$ correlations independently; thus ensuring that every node in the resulting distribution $P$ has degree $d$. This comes at a price, however: the analysis of the $\Omega(2^{d/2}\sqrt{n}/\eps^2)$ lower bound from~\cite{CanonneDKS20} crucially relied on independence across those components, and thus can no longer be extended to our case (where the $2^{d-1}$ distributions $P_\lambda$ share the same matching $M$, and thus we only have independence across components conditioned on $M$). Handling this requires entirely new ideas, and constitutes the core of our lower bound. In particular, from a technical point of view this requires us to handle the moment-generating-function of squares of Binomials (Lemma~\ref{lemma:MGF_Binomial_square_bound}), as well as that of (squares of) truncated Binomials (Lemma~\ref{lemma:MGF_Binomial_square_bound_with_min}). To do so, we develop in Section~\ref{ssec:tools:mgf} a range of results on Binomials and Multinomial distributions which we believe are of independent interest.

Finally, after establishing that $\Omega(2^{d/2}{n}/\eps^2)$ samples are necessary to distinguish the resulting ``mixture of trees'' $P$ from the uniform distribution $U$ (Lemma~\ref{theorem:lower_bound_uniformity_testing_on_Bayes_Net}), it remains to show that this implies our stronger statement on testing \emph{independence} (not just uniformity). To do so, we need to show that $P$ is not only far from $U$, but from \emph{every} product distribution: doing so is itself far from immediate, and is established in Lemma~\ref{lemma;distance:mixture:of:trees} by relating the distance from the mixture $P$ to every product distribution to the distance between distinct components of the mixture (Lemma~\ref{lemma:technical:TV_lower_bound_of_its_marginals}), and lower bounding those directly by analyzing the concentration properties of each component $P_\lambda$ of our construction (Lemma~\ref{lemma:technical_distance_lb_2xbinomials_conditional_tree}).

%% file: sec-preliminaries.tex
We use the standard asymptotic notation $O(\cdot)$, $\Omega(\cdot)$ $\Theta(\cdot)$, and write $\tilde{O}(\cdot)$ to omit polylogarithmic factors in the argument. Throughout, we identify probability distributions over discrete sets with their probability mass functions (pmf), and further denote by $U$ (resp., $U_d$) the uniform distribution on $\{0, 1\}^n$ (resp., $\{0, 1\}^d$). We also write $P^{\otimes m}$
for the $m$-fold product of a distribution $P$, that is, $P \otimes \cdots \otimes P$ (the distribution of a tuple of $m$ i.i.d.\ samples from $P$); and $[n]$ for the set $\{1,\dots,n\}$.\smallskip

\noindent\textbf{Bayesian networks.} Given a directed acyclic graph (DAG) over $n$ nodes, a probability distribution $P$ over $\{0,1\}^n$ is said to be \emph{Markov with respect to $G$} if $P$ factorizes according to $G$; we will also say that $P$ has \emph{structure} $G$. A DAG $G$ is said to have \new{in-}degree $d$, if every node has at most $d$ parents\new{; for convenience, we use \emph{degree-$d$} exclusively as ``in-degree $d$'' throughout the paper}; we will denote by $\Pi_i\subseteq [n]$ the set of parents of a node $i$. Finally, a distribution $P$ over $\{0,1\}^n$ is a \emph{degree-$d$ Bayes net} if $P$ is Markov w.r.t. some degree-$d$ DAG.\smallskip

\noindent\textbf{Distances between distributions.} Given two distributions $P,Q$ over the same (discrete) domain $\mathcal{X}$, the \emph{total variation distance} (TV) between $P$ and $Q$ is defined as
\begin{equation}
d_{\tmop{TV}}(P,Q) 
= \sup_{S\subseteq \mathcal{X}}(P(S)-Q(S)) 
= \frac{1}{2}  \sum_{x\in\mathcal{X}}^n | P(x)-Q(x) | \in [0,1]\,.
\end{equation}
While TV distance will be our main focus, we will also rely in our proofs on two other notions of distance between distributions: the \emph{Hellinger distance}, given by 
$
d_{\rm{}H} (P,Q) 
= \frac{1}{\sqrt{2}}\|\sqrt{P}-\sqrt{Q}\|_2
$, and the \emph{chi-squared divergence}, defined by 
\smash{$
d_{\chi^2}(P,Q) = \sum_{x\in\mathcal{X}} {(P(x)-Q(x))^2}/{Q(x)}
$}. TV distance, squared Hellinger distance, and $\chi^2$ divergence are all instances of $f$-divergences, and as such satisfy the data processing inequality; further, they are related by the following sequence of inequalities:
\begin{equation}
    \label{relation:distances}
    d_{\rm{}H}^2 (P,Q) \leq d_{\tmop{TV}} (P,Q) \leq \sqrt{2} d_{\rm{}H} (P,Q)
   \leqslant \sqrt{d_{\chi^2} (P,Q)} 
\end{equation}
\noindent\textbf{Tools from previous work.} 
We finally state results from the literature which we will rely upon.
\begin{corollary}[{\citet[Corollary 2.4]{daskalakis2016square}}]
\label{corollary:squared_hellinger_subadditivity}    
    Suppose $P$ and $Q$ are distributions on $\Sigma^n$ with common factorization structure
	\[ P (x) = P_{X_1} (x_1)  \prod_{i = 2}^n P_{X_i \mid X_{\Pi_i}} (x_i
	|x_{\Pi_i}), 
	\qquad  Q (x) = Q_{X_1} (x_1)  \prod_{i = 2}^n Q_{X_i \mid X_{\Pi_i}} (x_i
	|x_{\Pi_i}), \qquad x\in\Sigma^n
	\]
	where we assume the nodes are topologically ordered, and $\Pi_i$ is the
	set of parents of $i$. Then
	\[ d_{\rm{}H}^2 (P, Q) \leq d_{\rm{}H}^2 (P_{X_1}, Q_{X_1}) + d_{\rm{}H}^2 \mleft( P_{X_2,
		X_{\Pi_2}}, Q_{X_2, X_{\Pi_2}} \mright) + \cdots + d_{\rm{}H}^2 \mleft( P_{X_n,
		X_{\Pi_n}}, Q_{X_n, X_{\Pi_n}} \mright) . \]
	In particular, if $d_{\rm{}H}^2 (P, Q) \geq \eps$ then there exists some $i$ such that
	$ d_{\rm{}H}^2 \mleft( P_{X_i, X_{\Pi_i}}, Q_{X_i, X_{\Pi_i}} \mright) \geq
	\frac{\eps}{n} .
	$
\end{corollary}

\begin{lemma}[{\citet[Lemma~4]{AcharyaDK15}}] \label{lemma:product_learning_chi_square} There exists an efficient algorithm which, given 
	samples from a distribution $P$ over $[n_1]\times\cdots\times[n_d]$, outputs a product distribution $Q$ over $[n_1]\times\cdots\times[n_d]$ such that, if $P$ is a product distribution, then
	$
	    d_{\chi^2} (p, q) \leq \eps^2 . 
	$
	with probability at least 5/6. The algorithm uses 
	$
	O( (\sum_{\ell = 1}^d n_\ell)/{\eps^2} )
	$
	samples from $P$.
\end{lemma}

\begin{theorem}[{\citet[Theorem~1]{daskalakis2018distribution}}]
	\label{theorem:identity testing_chi_square_Hellinger} There exists an efficient algorithm which, given samples from a distribution $P$ and a known reference distribution $Q$ over $[n]$, as well as a distance parameter $\eps\in(0,1]$, distinguishes between the cases
	(i)~$d_{\chi^2} (P,Q) \leq \eps^2$ 
	and (ii)~$d_{\rm{}H} (P,Q) \geq \eps$ 
	with probability at least $5/6$. 
	The algorithm uses $O \mleft( {n^{1 / 2}}/{\eps^2} \mright)$
	samples from $P$.
\end{theorem}   

\begin{lemma}[{\citet[Proposition~1]{batu2001testing}}]\label{lemma:technical_approximated_lb_product_distributions}
  Let $P$ be a discrete distribution supported on $[n]
  \times [m]$, with marginals $P_1$ and $P_2$. Then, we have
  $
    \min_{Q_1, Q_2} d_{\tmop{TV}} (P, Q_1 \otimes Q_2) \geq
    \frac{1}{3} d_{\tmop{TV}} (P, P_1 \otimes P_2)\,,
  $
  where the minimum is taken over all distributions $Q_1$ over $[n]$ and $Q_2$ over $[m]$.
\end{lemma}

%
%
%
%
%
%

%% file: sec-upperbound.tex
In this section, we establish the upper bound part of Theorem~\ref{theo:main:informal}, stated below:
\begin{theorem}
  \label{theorem:independence_testing_ub}
  There exists an algorithm (Algorithm~\ref{alg:indep}) with the following guarantees: given a parameter $\eps\in(0,1]$ and sample access to an unknown degree-$d$ Bayes net $P$, the algorithm takes 
  $O(d^2 2^{d / 2} n \log(n)  / \eps^2)$
  samples from $P$, and distinguishes with probability at least $2/3$ between the cases (i)~$P$ is a product distribution, and (ii)~$P$ is at total variation distance at least $\eps$ from every product distribution.
\end{theorem}
\begin{algorithm2e}[htp]
\SetKwComment{Comment}{/* }{ */}
\SetKwInOut{Input}{Input}
\DontPrintSemicolon
\Input{Independent samples from a degree-$d$ Bayes net $P$ over $\{0,1\}^n$, $\eps\in(0,1]$}
Set $\delta \gets \frac{1}{3\binom{n}{d+1}}$, $\eps' \gets \frac{\eps}{\sqrt{2n}(1+\sqrt{d+1})}$, and $m \gets O\mleft(\frac{2^{d/2}}{\eps'^2}\log\frac{1}{\delta}\mright)$\;
Take a multiset $S$ of $m$ samples from $P$, where $\delta \gets \frac{1}{3\binom{n}{d+1}}$\;
\For{every subset $T\subseteq [n]$ of $d+1$ nodes }
{
    \Comment{We can generate one sample from $P_T$ given one from $P$}
     Use the $m$ samples from $S$ to generate $m$ i.i.d.\ samples from $P_T$\;
     Call Algorithm~\ref{alg:indep_general_Hellinger} with distance parameter $\eps'$ and failure probability $\delta$, using the $m$ samples from $P_T$
     \Comment{This is over $\{0,1\}^{d+1}$.}
}
\lIf{all $\binom{n}{d+1}$ tests accepted}
{
    \Return{accept}
} {
    \Return{reject}
}
\caption{Independence testing for degree-$d$ Bayes net}\label{alg:indep}
\end{algorithm2e}

As discussed in the introduction, the key components of our algorithm are performing the testing in Hellinger distance, in order to use the subadditivity result of~\citet{daskalakis2016square}; and using as subroutine an independence testing algorithm under Hellinger distance. As no optimal tester for this latter task was known prior to our work, our first step is to derive such an algorithm by adapting the ``testing-by-learning'' framework of~\citet{AcharyaDK15}. We emphasize that, in view of the relation $\frac{1}{2}d_{\tmop{TV}}^2\leq d_{\tmop{H}}^2 \leq d_{\tmop{TV}}$, this is a strictly stronger statement than the analogous known  $O(2^{n / 2}/ \eps^2)$-sample algorithm for testing independence under total variation distance~\citep{diakonikolas2016new,AcharyaDK15}, which only implies an $O(2^{n / 2}/ \eps^4)$ sample complexity for testing independence under Hellinger distance.
\begin{lemma}
  \label{lemma:hellinger_independence_tester}
  There exists an algorithm with the following guarantees: given a parameter $\eps\in(0,1]$ and sample access to an unknown distribution $P$ over $\{0,1\}^n$, the algorithm takes 
  $O(2^{n / 2}/ \eps^2)$
  samples from $P$, and distinguishes with probability at least $2/3$ between the cases (i)~$P$ is a product distribution, and (ii)~$P$ is at \emph{Hellinger} distance at least $\eps$ from every product distribution.
\end{lemma}
\begin{proof}%
	    We analyze Algorithm~\ref{alg:indep_general_Hellinger}:
		first we use the algorithm of Lemma~\ref{lemma:product_learning_chi_square} to learn $P$
		to $d_{\chi^2}$ distance $\eps^2$ \emph{as if} it was a product distribution, using $O \mleft( \frac{n}{\eps^2} \mright)$
		samples. Let $\hat{P}$ be the output of the learning algorithm. Note that since Lemma~\ref{lemma:product_learning_chi_square} guarantees \emph{proper} learning, $\hat{P}$ is a product distribution.
	
		We then want to check that $\hat{P}$ is indeed close to $P$ (in Hellinger distance), as it should if $P$ were indeed a product distribution. To do this, we use the algorithm of Theorem~\ref{theorem:identity testing_chi_square_Hellinger} on $P$, with reference distribution $\hat{P}$ and distance parameter $\eps$; and reject if, and only if, this algorithm rejects.
			
		By a union bound, since both algorithms are correct with probability at least $5/6$, both are simultaneously correct with probability at least $2/3$; we hereafter assume this is the case in our analysis. The total sample complexity is $O\mleft( {n}/{\eps^2} \mright) + O\mleft( {\sqrt{2^n}}/{\eps^2} \mright) = O\mleft( {2^{n/2}}/{\eps^2} \mright)$, as claimed. We now argue correctness.
		\begin{itemize}
		    \item\textit{Soundness.} Proof by contrapositive: if the algorithm accepts, this means, from the guarantees of Theorem~\ref{theorem:identity testing_chi_square_Hellinger}, that $d_{\rm H}(P,\hat{P}) \leq \eps$. Since $\hat{P}$ is a product distribution, we conclude that $P$ is $\eps$-close (in Hellinger distance) to being a product distribution.
		    
		    \item\textit{Correctness.} Assume that $P$ is a product distribution. Then, Lemma~\ref{lemma:product_learning_chi_square} ensures that $d_{\rm \chi^2}(P,\hat{P}) \leq \eps^2$; and thus, by Theorem~\ref{theorem:identity testing_chi_square_Hellinger} the second step will not reject, and the algorithm overall accepts.
		\end{itemize}
	Note that, by a standard amplification trick (independent repetition and majority vote), the probability of error can be reduced from $1/3$ to any $\delta\in(0,1)$ at the price of a $O(\log(1/\delta))$ factor in the number of samples.
\end{proof}

\begin{algorithm2e}[htp]
		\SetKwComment{Comment}{/* }{ */}
		\SetKwInOut{Input}{Input}
		\DontPrintSemicolon
		\Input{Independent samples from the target distribution over $\{0,1\}^n$, $\eps\in(0,1]$}
		Set $m_1 \gets O \mleft( n/\eps^2 \mright)$, and $m_2 \gets O\mleft( 2^{n/2}/\eps^2 \mright)$\;
		
		Take two multisets $S_1$ and $S_2$ of samples with size $m_1$, $m_2$ respectively from $P$\;
		
		$\hat{P} \gets$ call the algorithm of Lemma~\ref{lemma:product_learning_chi_square} with the samples from $S_1$\;
		
		Call the tester in Theorem~\ref{theorem:identity testing_chi_square_Hellinger} with the samples from $S_2$, to test $P$ with reference to $\hat{P}$ with distance parameter $\eps$\;
		
		\lIf{tester accepted}
		{
			\Return{accept}
		} {
			\Return{reject}
		}
		\caption{Hellinger independence testing for general distributions}\label{alg:indep_general_Hellinger}
\end{algorithm2e}

\new{We will next require the following lemma, whose proof is deferred to Appendix~\ref{proof:lemma:hellinger_projection_lb_bayes_net}:
	\begin{restatable}[]{lemma}{CorollaryHellingerProjectionLbProduct}
		\label{corollary:hellinger_projection_lb_product}
		Let $P$ be a distribution on $\{0,1\}^n$, and $P' = (\pi_1 P) \otimes \cdots
		\otimes (\pi_n P)$ be the product of marginals of $P$. Denoting by  $\mathcal{Q}$ the set of all product distributions on $\{ 0, 1 \}^n$,
		we have
		$
			\min_{Q \in \mathcal{Q}} d_{\rm{}H} (P, Q) \geq \frac{1}{1 + \sqrt{n}} d_{\rm{}H}(P, P') .
		$
	\end{restatable}
}
We are now ready to prove Theorem~\ref{theorem:independence_testing_ub}. As discussed in the introduction, the idea is the following: when $P$ is a product distribution, then the distribution induced by $P'$ on any subset of $d+1$ nodes is still a product distribution. However, by the above localization Corolloary \ref{corollary:squared_hellinger_subadditivity}, when $P$ is far from every product distribution, we can localize the farness between $P$ and $Q$ in one of the subsets of $d+1$ nodes. We also know that $P'$, the product of marginals $P$, is not too bad of an approximation to the closest $Q$ in the product space as suggested by Lemma \ref{corollary:hellinger_projection_lb_product}. Thus, testing independence using $P'$ as a proxy and paying an extra factor of $O(d)$ to compensate in accuracy suffice.

  \begin{proof}[Proof of Theorem~\ref{theorem:independence_testing_ub}]
    We analyze Algorithm~\ref{alg:indep}, denoting as in the algorithm by $P'$ the product of marginals of $P$, and setting 
    $\delta \eqdef \frac{1}{3\binom{n}{d+1}}$, $\eps' \eqdef \frac{\eps}{\sqrt{2n}(1+\sqrt{d+1})}$, and 
    \[
        m \eqdef O\mleft(\frac{2^{d/2}}{\eps'^2}\log\frac{1}{\delta}\mright)
        = O\mleft(\frac{2^{d/2}n}{\eps^2}\cdot d^2\log n\mright).
    \]
    The sample complexity is thus immediate; further, note that, as stated in the algorithm, given $m$ i.i.d.\ samples from $P$ and a fixed set $T\subseteq[n]$ of nodes, one can generate $m$ i.i.d.\ samples from $P_T$ by only keeping the relevant variables (those in $T$) of each sample from $P$.
    \begin{itemize}
        \item \textit{Completeness.} Assume $P$ is a product distribution. Then, $P_T$ is a product distribution for every choice of $T$, and each of the $\binom{n}{d+1}$ performs thus accepts with probability at least $1-\delta$ by Lemma~\ref{lemma:hellinger_independence_tester}. Thus, by a union bound, all tests simultaneously accept with probability at least $1-\binom{n}{d+1}\cdot \delta = 2/3$, and Algorithm~\ref{alg:indep} returns ``accept.''
        
        \item \textit{Soundness.} Assume now that $P$ is $\eps$-far (in total variation distance) from every product distribution over $\{0,1\}^n$. A fortiori, it is $\eps$-far from the product distribution $P'$, and thus we have
        \[
            d_{\rm{}H}  (P, P') \geq \frac{1}{\sqrt{2}}d_{\tmop{TV}}  (P, P') \geq \frac{\eps}{\sqrt{2}}
        \]
        By Corollary~\ref{corollary:squared_hellinger_subadditivity}, this means there exists some node $i\in[n]$ along with the set of its (at most) $d$ parents $\Pi_i$ such that, setting $T\eqdef \{i\}\cup\Pi_i$,
        \[
            d_{\rm{}H}^2(P_{T}, P'_T) \geq \frac{\eps^2}{2n}\,.
        \]
        Now, we can invoke our localization lemma, Lemma~\ref{corollary:hellinger_projection_lb_product}, to conclude that $P_T$ is not only far from $P'_T$, it is far from \emph{every} product distribution on $T$:
        \[
            \min_{Q\text{ product}} d_{\rm{}H}^2(P_{T}, Q) \geq \frac{\eps^2}{2n(1+\sqrt{d+1})^2} = \eps'^2\,.
        \]
        Thus, when this particular set $T$ of $d+1$ nodes is encountered by the algorithm, the corresponding independence test will reject with probability at least $1-\delta$ by Lemma~\ref{lemma:hellinger_independence_tester}, and the overall algorithm will thus reject with probability at least $1-\delta$.
    \end{itemize}
    This concludes the proof: the sample complexity is $O(d^2 2^{d / 2} n \log(n)  / \eps^2)$ as claimed, and the tester is correct in both cases with probability at least $2/3$.
  \end{proof}

%% file: sec-lowerbound.tex
In this section, we prove our main technical result: a lower bound on the sample complexity of testing whether a degree-$d$ Bayes net is a product distribution.
\begin{theorem}
    \label{theo:main:lb}
  Let $d \leq c\cdot \log n$, where $c>0$ is a sufficiently small absolute constant. Then, testing whether an arbitrary
  degree-$d$ Bayes net over $\{0, 1\}^n$ is a product distribution or is
  $\eps$-far from every product distribution requires $\Omega (2^{d / 2} n /
  \eps^2)$ samples.
\end{theorem}
This theorem considerably generalizes the lower bound of $\Omega(n/\eps^2)$ established in~\cite{CanonneDKS20} for the case of trees\footnote{\new{While this can be better visualized as a forest, we use the term “tree” to refer to the structure of degree-$1$ Bayes nets. Technically, it still factorizes as a Bayesian path (tree), but not all edges are necessary.}} ($d=1$). To establish our result, we build upon (and considerably extend) their analysis; in particular, we will rely on the following ``mixture of trees'' construction, which can be seen as a careful mixture of ($2^{d-1}$ of) the hard instances from the lower bound of~\cite[Theorem 14]{CanonneDKS20}. \new{Figure \ref{fig:lower_bound_construction} shows an illustrative example of our lower bound constructions.}

\paragraph{Notation.}
  Throughout, for given $n, d, \eps$, we let $N\eqdef n-d+1$ (without loss of generality assumed to be even), $D \eqdef 2^{d-1}$, $\delta \eqdef \frac{\eps}{\sqrt{n}}$, and 
  \[
  z_0 = \frac{1 + 4 \delta}{1 - 4 \delta}, \quad
  z_1 = \frac{1 + (4
  \delta)^2}{1 - (4 \delta)^2}, \quad
  z_2 = \frac{1 + (4 \delta)^4}{1 - (4
  \delta)^4}.
  \]

\begin{definition}[Mixture of Trees]
    \label{def:mixture:trees}
Given parameters $0\leq d \leq n$ and $\delta\in(0,1]$, we define the probability distribution $\mathcal{D}_{n,d,\delta}$ over degree-$d$ Bayes nets by the following process.
\begin{enumerate}
    \item Choose a perfect matching $\lambda$ of $[N]$ uniformly at random (where $N=n-d+1$), i.e., a set of $N/2$ disjoint pairs;
    \item Draw i.i.d.\ $\mu_1,\dots, \mu_D$ uniformly at random in $\{0,1\}^{N/2}$;
    \item For $\ell\in[D]$, let $p_\ell$ be the distribution over $\{0,1\}^N$ defined as the Bayes net with tree structure $\lambda$, such that if $\lambda_k = (i,j)\in\lambda$ then the corresponding covariance between variables $X_i,X_j$ is 
    \[
    \operatorname{Cov}(X_i, X_j) = (- 1)^{\mu_{\ell, k}}
\delta
    \]
    \item\label{def:mixture:trees:item4} Let the resulting distribution $P_{\lambda,\mu}$ over $\{0,1\}^n$ be
    \[
        P_{\lambda,\mu}(x) = \frac{1}{D}\sum_{\ell=1}^D p_\ell(x_{d},\dots,x_{n}) \mathbbm{1}[\iota(x_1,\dots,x_{d-1})=\ell]
    \]
    where $\iota\colon \{0,1\}^{d-1} \to [D]$ is the indexing function, mapping the binary representation to the corresponding number. 
\end{enumerate}
That is, $\mathcal{D}_{n,d,\delta}$ is the uniform distribution over the set of degree-$d$ Bayes nets where the first $d-1$ coordinates form a ``pointer'' to one of the $2^{d-1}$ tree Bayes nets sharing the same tree structure (the matching $\lambda$), but with independently chosen covariance parameters (the $D$ parameters $\mu_1,\dots, \mu_D$).
\end{definition}

\new{
  \begin{figure}[h!]
  \centering
  \includegraphics[width=1.0\linewidth]{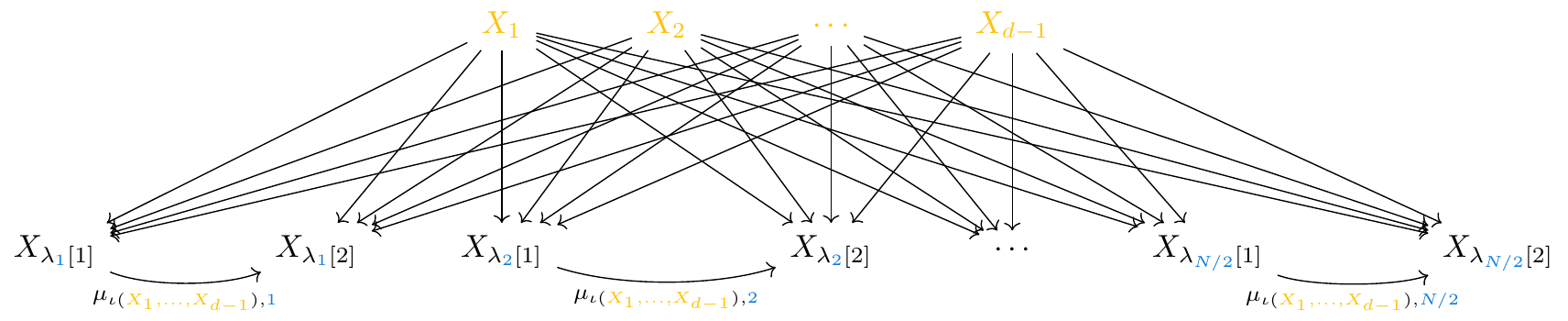}
    \caption{\small
      A depiction of our lower bound construction, where the arrows are the edges of the underlying Bayes net. The first $(d-1)$ nodes can be seen as a $(d-1)$-bit string that acts as a pointer to one of the distributions, i.e., the value of the first $d$ nodes encodes which distribution the rest of the nodes are on.
      We write $\lambda_k[1] \eqdef i$ and $\lambda_k[2]\eqdef j$ when $\lambda_k=(i,j)$. There are two sources of randomness in the construction. The first one is the random matching on the $N/2$ pairs represented in graph in the form of
      $ (\lambda_1 [1], \lambda_1 [2]), \ldots, (\lambda_{N / 2} [1], \lambda_{N / 2} [2]) $, which is independent of the particular configuration $X_1, \ldots, X_d$ takes. The second one is the covariance parameters between each pair; these parameters are randomly and independently chosen for each configuration (in total $2^d$ of them) of the first $d$ nodes. Color (yellow and blue) marks how the value of covariance parameter $\mu_{\ell,k}$ in each pair depends on the first $d$ nodes' configurations as well as the matching parameter $k$.
    }
    \label{fig:lower_bound_construction}
  \end{figure}
}

With this construction in hand, Theorem~\ref{theo:main:lb} will follow from the next two lemmas:
\begin{lemma}[Indistinguishability]
  \label{theorem:lower_bound_uniformity_testing_on_Bayes_Net}
  There exist absolute constants $c,C>0$ such that the following holds. 
  For $\Omega(1) \leq d \leq c\log n$, no $m$-sample algorithm can distinguish with probability at least $2/3$ between a (randomly chosen) mixture of trees $P\sim \mathcal{D}_{n,d,\eps/\sqrt{n}}$ and the uniform distribution $U$ over $\{0,1\}^n$, unless $m \geq C\cdot \frac{2^{d / 2} n}{\eps^2}$.
\end{lemma}
\begin{lemma}[Distance from product]
  \label{lemma;distance:mixture:of:trees}
  Fix any $0\leq d\leq \frac{n}{2}$ and $\eps\in(0,1]$. With probability at least $9/10$ over the choice of $P\sim\mathcal{D}_{n,d,\eps/\sqrt{n}}$, $P$ is $\Omega(\eps)$-far from every product distribution on $\{0,1\}^n$.
\end{lemma}
Note that this guarantees farness from \emph{every} product distribution, not just from the uniform distribution. 
We first establish Lemma~\ref{theorem:lower_bound_uniformity_testing_on_Bayes_Net} in the next subsection, before proving Lemma~\ref{lemma;distance:mixture:of:trees} in Section~\ref{ssec:mixturetree:farness}. Throughout, we fix $n$, $d$, and $\eps\in(0,1]$.    
\subsection{Sample Complexity to distinguish from Uniform (Lemma~\ref{theorem:lower_bound_uniformity_testing_on_Bayes_Net})}
    \label{ssec:distinguish:uniform}
%
\noindent\textbf{Proof sketch.}
\new{Following the original analogous analysis, we first use Ingster's method to upper bound square $\tmop{TV}$ distance in \eqref{eq:Ingster:restated} and after a series of algebraic calculations, we arrive at \eqref{eq:multigraph_cycle_expectation}. This is where we improve upon the original analysis by substituting a tighter upper bound \eqref{fact:intermedia_upper_bound_on_multigraph_cycle_term:restated}. This leads us to some of the most technical portions of the paper. To upper bound the inner expectation of~\eqref{eq:boundexpect:partial}, we use Lemma \ref{lemma:average_of_product_with_binomial_raise_by_m:restated_improve_attempt} to bound the unwieldy expression (average over $z_j^{\kappa_{j,i}}$ and raised by $m$) by an expectation of a multinomial variable $\orgvec{\alpha}$ in the expression \smash{$\textstyle\prod_{i = 1}^D \prod^K_{j = 1} \cosh (2 \alpha_i \delta_j)$} in \eqref{eq:boundinnerexpect:partial_improve_attempt}; then, with some careful analysis and some tools on MGF (Moment Generating Function) of Binomial and Multinomial from Section \ref{ssec:tools:mgf}, we can upper bound the latter expression in \eqref{equation:apply_MGF_min_square_bin}, which finally gives us the sample complexity lower bound.}

\noindent\textbf{Details.}
We here proceed with the proof of Lemma~\ref{theorem:lower_bound_uniformity_testing_on_Bayes_Net}, starting with some convenient notations; some of the technical lemmas and facts used here are stated and proven in Section~\ref{ssec:tools:mgf}. 
Let $\theta = (\lambda, \mu_1, \ldots, \mu_{2^{d-1}})$, where each $\mu_k\in\{\pm 1\}^{N/2}$, and let $P_{\theta}$ be
the distribution for the mixture of trees construction from Definition~\ref{def:mixture:trees}. The
following denotes the \emph{matching count} between $(\lambda, \mu)$ and $x$ as the quantity
\[ 
c (\lambda, \mu, x) \eqdef  \mleft| \mleft\{ (i, j) \in \{d, \dots, n\}^2
    : \exists 1\leq k \leq \tfrac{N}{2}, \lambda_k = (i, j) \infixand (- 1)^{x_i + x_j}
    = (- 1)^{\mu_{\iota(x),k}} \mright\} \mright| . 
\]
We will also introduce an analogous quantity with an ``offset'', for $x_{\tmop{ch}} = (x_{d},
\ldots, x_n)$, referring exclusively to the child nodes of $x$ (i.e., the last $N$ nodes, which are the ``children'' of the first $d-1$ ``pointer nodes'' in our construction),
\begin{align}
&c_{\tmop{ch}} (\lambda, \mu, x_{\tmop{ch}}) \eqdef\label{def:ch}\\ 
  &\mleft| \mleft\{ (i, j) \in \{d, \dots, n\}^2 : \exists 1\leq k \leq \tfrac{N}{2},
  \lambda_k = (i - d + 1, j - d + 1) \infixand (- 1)^{x_i + x_j} = (- 1)^{\mu_{\iota(x),k}}
  \mright\} \mright| \,. \notag
\end{align}
To denote the parameters of the ``mixture of trees,'' we write $\theta_i \eqdef (\lambda, \mu_i)$ (for $i\in[D]$), recalling that the matching parameter $\lambda$ is common to all $D$ tree components. Since each $\mu_i$
corresponds to one of the values of $(x_1, \ldots, x_{d-1}) \in \{0, 1\}^{d-1}$, we as before use 
$\iota\colon \{0,1\}^{d-1} \to [D]$ to denote the indexing function (so that, for instance, $\iota(x_1 =
\cdots = x_{d-1} = 0) = 0$). We finally introduce three more quantities, related to the matching and orientations parameters across the $D$ components of the mixture:
\begin{align*}
A_{\theta_i, \theta_i'} &\eqdef  \{(s, t) \in \{1,\dots,N / 2\}^2 :
    \lambda [s] = \lambda' [t], \mu_i [s] = \mu_i' [t] \} \tag{common
    pairs, same orientation} \\
B_{\theta_i, \theta_i'} &\eqdef  \{(s, t) \in \{1,\dots,N / 2\}^2 :
    \lambda [s] = \lambda' [t], \mu_i [s] \neq \mu_i' [t] \}
    \tag{common pairs, different orientation} \\
C_{\theta_i, \theta_i'} & \eqdef 
    (\lambda \cup \lambda') \setminus (A \cup B)  \tag{pairs unique to
    $\theta$ or $\theta'$} 
\end{align*}
For ease of notation, we define $A_i \eqdef A_{\theta_i, \theta_i'}, B_i \eqdef
B_{\theta_i, \theta_i'} \infixand C_i \eqdef C_{\theta_i, \theta_i'}$; and note that $C_i=C_1$, as it only depends on $\lambda$ (not on the orientation $\mu_i$).\label{def:AB:lb}

  To prove the indistinguishability, we will bound the squared total variation distance (or equivalently,
squared $\ell_1$) distance between the distributions of $m$ samples from (the uniform mixture of) $P_{\theta}$ and $U$ by a small constant; that is, between $Q \eqdef \mathbb{E}_{\theta}
[P_{\theta}^{\otimes m}]$ and $U^{\otimes m}$. From Ingster's method (see, e.g., \cite[Lemma~III.8.]{acharya2018inference}), by using chi-square divergence as an intermediate step we get
\begin{align}
  \norm{Q - U^{\otimes m} }_1^2 & \leq d_{\chi^2} (Q, U^{\otimes m})  
  = \mathbb{E}_{\theta, \theta'} [(1 + \tau (\theta, \theta'))^m] - 1 \label{eq:Ingster:restated}
  \text{,}
\end{align}
where $\tau (\theta, \theta') \eqdef \mathbb{E}_{x \sim U} \mleft[ \mleft(
\frac{P_{\theta} (x) - U (x)}{U (x)} \mright) \mleft( \frac{P_{\theta'} (x) - U
(x)}{U (x)} \mright) \mright]$. 
In order to get a handle on this quantity $\tau (\theta, \theta')$, we start by writing the expression for the density $P_{\theta}$ (for a given parameter $\theta$ of the mixture of trees). For any  $x\in\{0,1\}^n$, recalling Item~\ref{def:mixture:trees:item4} of Definition~\ref{def:mixture:trees},
\begin{eqnarray*}
  P_{\theta} (x) & = & P_{\theta} (x_{d}, \ldots, x_n \mid x_1, \ldots, x_{d-1}) U_{d-1}
  (x_1, \ldots, x_{d-1})\\
  & = & \frac{1}{2^{d-1}} P_{\lambda, \mu_{\iota(x)}} (x_{d}, \ldots, x_n) \\
  & = & \frac{1}{2^{d-1}} \cdot \frac{1}{2^N}  (1 + 4 \delta)^{c (\lambda, \mu_{\iota(x)}, x)}  (1 -
  4 \delta)^{\frac{N}{2} - c (\lambda, \mu_{\iota(x)}, x)}\\
  & = & \frac{1}{2^n}  (1 + 4 \delta)^{c (\lambda, \mu_{\iota(x)}, x)}  (1 - 4
  \delta)^{\frac{N}{2} - c (\lambda, \mu_{\iota(x)}, x)} .
\end{eqnarray*}
Substituting this in the definition of $\tau$, we get
\begin{align}
  \tau (\theta,\theta')
  & = \mathbb{E}_{x \sim U}  \mleft[ \mleft( \frac{P_{\theta} (x)}{U (x)} - 1
  \mright)  \mleft( \frac{P_{\theta'} (x)}{U (x)} - 1 \mright) \mright]\notag\\
  & = \mathbb{E}_{x \sim U}  \mleft[ \mleft( (1 - 4 \delta)^{\frac{N}{2}}
  \mleft( \frac{1 + 4 \delta}{1 - 4 \delta} \mright)^{c (\lambda, \mu_{\iota(x)},
  x)} - 1 \mright)  \mleft( (1 - 4 \delta)^{\frac{N}{2}} \mleft( \frac{1 + 4
  \delta}{1 - 4 \delta} \mright)^{c (\lambda', \mu_{\iota(x)}', x)} - 1 \mright)
  \mright] \notag\\
  &= 1 +(1-4\delta)^N\mathbb{E}_{x \sim U} \mleft[ z_0^{c (\lambda, \mu_{\iota(x)}, x) + c (\lambda', \mu_{\iota(x)}', x)} \mright] \notag\\
  &\quad- (1 - 4 \delta)^{\frac{N}{2}}\mathbb{E}_{x \sim U} \mleft[ z_0^{c (\lambda, \mu_{\iota(x)}, x)}
  \mright] -(1 - 4 \delta)^{\frac{N}{2}} \mathbb{E}_{x \sim U} \mleft[ z_0^{c (\lambda', \mu_{\iota(x)}', x)}
  \mright] \,, \label{eq:calculation:tau}
\end{align}
where $z_0 \eqdef \frac{1 + 4
\delta}{1 - 4 \delta}$. 
As $x \sim U_N$, for fixed $\lambda,\mu$ we have that $c_{\tmop{ch}} (\lambda, \mu, x)$ follows a $\tmop{Bin}
\mleft( \frac{N}{2}, \frac{1}{2} \mright)$ distribution; recalling the expression of the Binomial distribution's probability
generating function, we then have
\begin{eqnarray*}
  (1 - 4 \delta)^{\frac{N}{2}} \mathbb{E}_{x \sim U} [z_0^{c (\lambda, \mu_{\iota
  (x)}, x)}] & = & (1 - 4 \delta)^{\frac{N}{2}} \mathbb{E}_{\orgvec{x}_{1} \sim U_{d-1}}
  \mleft[ \mathbb{E}_{\orgvec{x}_2 \sim U_N} \mleft[ z_0^{c_{\tmop{ch}} (\lambda, \mu_{\iota
  (\orgvec{x}_1)}, \orgvec{x}_2)} \mright] \mright]\\
  & = & (1 - 4 \delta)^{\frac{N}{2}}  \frac{1}{D} \Bigg\{ \sum_{i = 1}^D
  \mathbb{E}_{m \sim \tmop{Bin} \mleft( \frac{N}{2}, \frac{1}{2} \mright)}
  [z_0^m] \Bigg\} \\
  & = & (1 - 4 \delta)^{\frac{N}{2}}  \mleft(
  \frac{1 + z_0}{2} \mright)^{\frac{N}{2}} = 1 \text{.}
\end{eqnarray*}
Using this to simplify the last two terms of~\eqref{eq:calculation:tau}, we obtain
\begin{align}
  1 + \tau (\theta, \theta') 
  & = (1 - 4 \delta)^N \mathbb{E}_{x \sim U}  [z_0^{c (\lambda, \mu_{\iota(x)}, x) + c (\lambda', \mu_{\iota(x)}', x)}] \notag\\
  & = \frac{1}{D}  \Bigg\{ \sum_{i = 1}^D (1 - 4 \delta)^N \mathbb{E}_{x \sim U_N}  [ z_0^{c_{\tmop{ch}} (\lambda, \mu_i, x) + c_{\tmop{ch}}
  (\lambda', \mu_i', x)}] \Bigg\} \notag\\
  & = \frac{1}{D}  \Bigg\{ \sum_{i = 1}^D (1 - 4 \delta)^N \cdot z_0^{|B_i |}
  \mathbb{E}_{\alpha \sim \tmop{Bin} \mleft( |A_i |, \frac{1}{2} \mright)}
  [z_0^{2 \alpha}]  \prod_{\sigma_i : | \sigma_i | \geq 4}
  \mathbb{E}_{\alpha \sim \mathcal{B} (\sigma_i)}
  [z_0^{\alpha}] \Bigg\} \notag\\
  & = \frac{1}{D}  \Bigg\{ \sum_{i = 1}^D (1 - 4 \delta)^N \cdot z_0^{|B_i |}
  \mleft( \frac{1 + z_0^2}{2} \mright)^{|A_i |}  \prod_{\sigma_i : | \sigma_i |
  \geq 4} \mathbb{E}_{\alpha \sim \mathcal{B} (\sigma_i)}
  [z_0^{\alpha}] \Bigg\} , \label{eq:1+tau:restated}
\end{align}
where the product is taken over all cycles in the multigraph $G_{\theta,\theta'}$ induced by the two matchings; and, given a cycle $\sigma$, $\mathcal{B} (\sigma)$ is the probability distribution defined as follows. Say that a cycle $\sigma$ is \emph{even} (resp., \emph{odd}) if the number of edges with weight $1$ along $\sigma$ is even (resp., odd); that is, a cycle is even or odd depending on whether the number of negatively correlated pairs along the cycle is an even or odd number. 
If $\sigma$ is an \emph{even cycle}, then $\mathcal{B} (\sigma)$ is a Binomial with parameters $|\sigma|$ and $1/2$, conditioned on taking even values. Similarly, if $\sigma$ is an  \emph{odd cycle}, $\mathcal{B} (\sigma)$ is a Binomial with parameters $|\sigma|$ and $1/2$, conditioned on taking odd values. 
It follows that $\mathbb{E}_{\alpha \sim
  \mathcal{B} (\sigma)} [z_0^{\alpha}]$ is given by the following expression.
  \[ \mathbb{E}_{\alpha \sim \mathcal{B} (\sigma)}
      [z_0^{\alpha}] = \mleft\{\begin{array}{ll}
        \mathbb{E}_{\alpha \sim \tmop{Bin} \mleft( | \sigma |, \frac{1}{2}
        \mright)} \mleft[ z^{\alpha} \mid \alpha \text{ even} \mright] = \frac{(1 +
        z_0)^{| \sigma |} + (1 - z_0)^{| \sigma |}}{2^{| \sigma |}} & \text{,
        if } \sigma \text{ is even}\\
        \mathbb{E}_{\alpha \sim \tmop{Bin} \mleft( | \sigma |, \frac{1}{2}
        \mright)} \mleft[ z^{\alpha} \mid \alpha \text{ odd} \mright] = \frac{(1 +
        z_0)^{| \sigma |} - (1 - z_0)^{| \sigma |}}{2^{| \sigma |}} & \text{,
        if } \sigma \text{ is odd}
      \end{array}\mright. \text{.} \]
  Denote
  \begin{align*}
      \mathcal{S}_{e, i} &\eqdef \mathcal{S}_e (\theta_i, \theta_i') = \mleft\{ \sigma
      \in \tmop{cycle} (\theta_i, \theta_i') : | \sigma | \geq 4, \sigma
      \text{ is even} \mright\}, \\
      \mathcal{S}_{o, i} &\eqdef\mathcal{S}_o (\theta_i, \theta_i') = \mleft\{ \sigma
      \in \tmop{cycle} (\theta_i, \theta_i') : | \sigma | \geq 4, \sigma
      \text{ is odd} \mright\} .
  \end{align*} 
  We will often drop $\lambda, \mu$ or $i$, when clear from context. We expand $\mathbb{E}_{\alpha \sim \mathcal{B} (\sigma)}
[z_0^{\alpha}]$ as follows: 
\begin{eqnarray}
  \prod_{\sigma : | \sigma | \geq 4} \mathbb{E}_{\alpha \sim
  \mathcal{B} (\sigma)} [z_0^{\alpha}]
  & = & \prod_{\sigma : | \sigma | \geq 4, \tmop{even}} \frac{(1 +
  z_0)^{| \sigma |} + (1 - z_0)^{| \sigma |}}{2^{| \sigma |}}  \prod_{\sigma :
  | \sigma | \geq 4, \tmop{odd}} \frac{(1 + z_0)^{| \sigma |} - (1 -
  z_0)^{| \sigma |}}{2^{| \sigma |}} \notag\\
  & = & \prod_{\substack{\sigma : | \sigma | \geq 4\\ \tmop{even}}} \frac{(1 +
  z_0)^{| \sigma |}}{2^{| \sigma |}} \!\! \mleft( 1 + \mleft( \frac{1 - z_0}{1 +
  z_0} \mright)^{| \sigma |} \mright)
    \prod_{\substack{\sigma : | \sigma | \geq 4 \notag\\\tmop{odd}}} \!\!\frac{(1 + z_0)^{|
  \sigma |}}{2^{| \sigma |}}  \mleft( 1 - \mleft( \frac{1 - z_0}{1 + z_0}
  \mright)^{| \sigma |} \mright) \notag\\
  & = & \prod_{\sigma : | \sigma | \geq 4} \frac{(1 + z_0)^{| \sigma
  |}}{2^{| \sigma |}}  \prod_{\sigma \in \mathcal{S}_e} (1 + (- 4 \delta)^{|
  \sigma |})  \prod_{\sigma \in \mathcal{S}_o} (1 - (- 4 \delta)^{| \sigma
  |}) \notag\\
  & = & \frac{(1 + z_0)^{\sum_{\sigma: |\sigma|\geq 4} |\sigma|}}{2^{\sum_{\sigma: |\sigma|\geq 4} |\sigma|}}  \prod_{\sigma \in \mathcal{S}_e} (1 + (- 4 \delta)^{|
  \sigma |})  \prod_{\sigma \in \mathcal{S}_o} (1 - (- 4 \delta)^{| \sigma
  |}) \notag\\
  & = & \frac{(1 + z_0)^{|C|}}{2^{|C|}}  \prod_{\sigma \in \mathcal{S}_e} (1
  + (- 4 \delta)^{| \sigma |})  \prod_{\sigma \in \mathcal{S}_o} (1 - (- 4
  \delta)^{| \sigma |}) \label{eq:multigraph_cycle_expectation}
\end{eqnarray}
where for the last equality we used that $\sum_{\sigma: |\sigma|\geq 4} |\sigma| = |C|$. We now improve upon the analogous analysis from~\cite[Claim~12]{CanonneDKS20} to
obtain a better upper bound for the remaining terms; indeed, the bound they derived is $e^{O(\eps^4/n)}$, which was enough for their purposes but not ours (since it does not feature any dependence on $d$). Let
$z_1 \eqdef \frac{1 + (4 \delta)^2}{1 - (4 \delta)^2}$. In view of using the above expression to bound~\eqref{eq:1+tau:restated}, we first simplify (part of) the summands of~\eqref{eq:1+tau:restated} by using the fact that $2|A_i|+2|B_i|+|C_i|=N$ for all $i$, and following the same computations as in~\citet{CanonneDKS20}:
\begin{align*}
  (1 - 4 \delta)^N  z_0^{|B_i |} &
  \mleft( \frac{1 + z_0^2}{2} \mright)^{|A_i |}  \frac{(1 + z_0)^{|C_i|}}{2^{|C_i|}} \\
  &= (1 - 4 \delta)^N  z_0^{|B_i |} \mleft( \frac{1 + z_0^2}{2}
  \mright)^{|A_i |} \frac{(1 + z_0)^{|C_i |}}{2^{|C_i |}}\\
  & = ((1 - 4 \delta)^2 z_0)^{|B_i |}  \mleft( (1 -
  4 \delta)^2  \frac{1 + z_0^2}{2} \mright)^{|A_i |}
  \underbrace{\mleft( (1 - 4 \delta) \frac{1 + z_0}{2} \mright)^{|C_i |}}_{=
  1}\\
  & =  (1 - (4 \delta)^2)^{|B_i |}  (1 + (4 \delta)^2)^{|A_i |}\\
  & =  (1 - (4 \delta)^2)^{|A_i | + |B_i |} z_1^{|A_i |}
  =  (1 - (4 \delta)^2)^{|A_1 | + |B_1 |} z_1^{|A_i |} \,,
\end{align*}
where the last equality uses the fact that the sum $|A_i | + |B_i |$ only depends on the matchings $\lambda,\lambda'$ (not the orientations $\mu_i, \mu_i'$), and thus is independent of $i$. 
Plugging this simplification into~\eqref{eq:1+tau:restated}, and letting $R \eqdef |A_1|+|B_1| \leq N/2$ for convenience, we get
\begin{align*}
  1 + \tau (\theta, \theta')
  & = \frac{1}{D} \Bigg\{ \sum_{i = 1}^D (1 - 4 \delta)^N \cdot z_0^{|B_i |}
  \mleft( \frac{1 + z_0^2}{2} \mright)^{|A_i |}  \prod_{\sigma_i}
  \mathbb{E}_{\alpha \sim \mathcal{B}(\sigma_i)} [z_0^{\alpha}] \Bigg\}\\
  & = (1 - (4 \delta)^2)^{R} \cdot  \frac{1}{D}  \Bigg\{ \sum_{i = 1}^D 
  z_1^{|A_i |}  \prod_{\sigma \in \mathcal{S}_{e, i}} (1 + (- 4
  \delta)^{| \sigma |}) \prod_{\sigma \in \mathcal{S}_{o, i}} (1 - (- 4
  \delta)^{| \sigma |}) \Bigg\} \text{.}
\end{align*}
Next, we compute the expectation after raising the above to the power $m$.
\begin{align}
  \mathbb{E}_{\theta, \theta'} &[(1 + \tau (\theta, \theta'))^m] \notag\\
  & = \mathbb{E}_{\theta, \theta'} \mleft[ \mleft( (1 - (4 \delta)^2)^{R} 
  \frac{1}{D} \Bigg\{ \sum_{i = 1}^D  z_1^{|A_i |}  \prod_{\sigma \in
  \mathcal{S}_{e, i}} (1 + (- 4 \delta)^{| \sigma |}) \prod_{\sigma \in
  \mathcal{S}_{o, i}} (1 - (- 4 \delta)^{| \sigma |}) \Bigg\} \mright)^m
  \mright] \notag\\
  & = \mathbb{E}_{\lambda, \lambda'}  \mleft[ (1 - (4 \delta)^2)^{mR}
  \mathbb{E}_{\orgvec{\mu}, \orgvec{\mu}'} \mleft[ \mleft( \frac{1}{D} \Bigg\{ \sum_{i = 1}^D 
  z_1^{|A_i |}  \prod_{\sigma \in \mathcal{S}_{e, i}} (1 + (- 4 \delta)^{|
  \sigma |}) \prod_{\sigma \in \mathcal{S}_{o, i}} (1 - (- 4 \delta)^{| \sigma
  |}) \Bigg\} \mright)^m \mright] \mright]. \label{eq:boundexpect:partial}
\end{align}
The quantity inside the inner expectation is quite unwieldy; to proceed, we will rely on the following identity, which lets us bound the two product terms:
\begin{equation}
\label{fact:intermedia_upper_bound_on_multigraph_cycle_term:restated}
  \prod_{\sigma \in \mathcal{S}_e} (1 + (- 4 \delta)^{| \sigma |}) 
  \prod_{\sigma \in \mathcal{S}_o} (1 - (- 4 \delta)^{| \sigma |})
  \leq e^{c'\frac{\eps^5}{N^{3 / 2}}}\!\!\!\!\!\!  \prod_{\sigma \in
  \mathcal{S}_e : | \sigma | = 4} \!\!\!\!\!\!(1 + (- 4 \delta)^{| \sigma |}) 
  \!\!\!\!\!\!\prod_{\sigma \in \mathcal{S}_o : | \sigma | = 4} \!\!\!\!\!\!(1 - (- 4 \delta)^{|
  \sigma |})\,,
\end{equation}
for some absolute constant $c'>0$. We defer the proof of this inequality to Appendix~\ref{app:fact:intermedia_upper_bound_on_multigraph_cycle_term:proof}, and proceed assuming it. Note that as long as $D = {O}(n/\eps^6)$, we will have $e^{c'\cdot  m\eps^5 / N^{3 / 2}} \leq e^{128 m \eps^2 / (\sqrt{D} n)}$,\footnote{{As per the condition set in Lemma \ref{lemma:multinomial_truncated_MGF}, we will from now on assume that $n / \varepsilon^2 \geqslant n \geqslant 40 D$, which gives us $N^{3 / 2} \geqslant \left( n/2 \right)^{3 / 2} \geqslant \frac{n}{2}
\cdot (20 D)^{1 / 2} \geqslant 2 n \sqrt{D}$; and some more calculations give us $c'\frac{m \varepsilon^2}{N^{3 / 2}} \leqslant \frac{128
m \varepsilon^2}{\sqrt{D} n}$.}} and this restriction on $D$ is satisfied for the regime of parameters considered in our lower bound, $d = O(\log n)$.

Fix a pair $\lambda, \lambda'$; we have that the $|A_i |$'s are i.i.d.\ $\tmop{Bin}
(R, 1 / 2)$ random variables. We now introduce \[
R' \eqdef |\{\sigma_1 : | \sigma_1 | = 4\}| =
|\{\sigma_i : | \sigma_i | = 4\}|,
\] which is the random variable denoting how many cycles have length exactly 4. In particular, we have $R'
\leq \frac{N}{4}$, since $\sum_{\sigma: |\sigma|\geq 4} \sigma = |C | \leq
N$; more specifically, we have $R' \leq \frac{N - 2 R}{4} \leq \frac{N}{4}$.     
Further, define $\kappa_i$ as the number of cycles of length 4 which have an even total number of negative correlations; that is, the number of cycles $\sigma$ such that  $\mu_i, \mu_i'$ impose either 0, 2, or 4 negatively correlated pairs along that cycle.

Since $\mu, \mu'$ are uniformly distributed, being odd or even each has probability $1 / 2$, and thus $\kappa_i \sim \tmop{Bin} \mleft(
R', \frac{1}{2} \mright)$. Moreover, while $\kappa_i$ and $A_i$ both depend on $\mu_i, \mu_i'$, they by definition depend on disjoint subsets of those two random variables: thus, because each correlation parameter is chosen independently, we have that 
$\kappa_i$ and $A_i$ are independent conditioned on $(R, R')$. Now, recalling our setting of $z_2 = \frac{1 + (4 \delta)^4}{1
- (4 \delta)^4}$ and fixing a realization of $R,R'$, we have

 \begin{align}
  \mathbb{E}_{\orgvec{\mu}, \orgvec{\mu}'} & \mleft[ \mleft( \frac{1}{D} 
  \sum_{i = 1}^D z_1^{|A_i |}  \prod_{\sigma \in \mathcal{S}_e (i) : | \sigma
  | = 4} (1 + (4 \delta)^4) \prod_{\sigma_i \in \mathcal{S}_o (i) : | \sigma |
  = 4} (1 - (4 \delta)^4) \mright)^m \mright] \nonumber\\
  & =\mathbb{E}_{\orgvec{\mu}, \orgvec{\mu}'} \mleft[ \mleft( \frac{1}{D} 
  \sum_{i = 1}^D z_1^{|A_i |} (1 + (4 \delta)^4)^{\kappa_i} (1 - (4
  \delta)^4)^{R' - \kappa_i} \mright)^m \mright] \nonumber\\
  & = (1 - (4 \delta)^4)^{mR'} \mathbb{E}_{\orgvec{\mu}, \orgvec{\mu}'}
  \mleft[ \mleft( \frac{1}{D}  \sum_{i = 1}^D z_1^{|A_i |} z_2^{\kappa_i}
  \mright)^m \mright] \nonumber\\
  & \leq (1 - (4 \delta)^4)^{mR'} z_1^{\frac{mR}{2}} z_2^{\frac{mR'}{2}}
  \mathbb{E}_{\orgvec{\alpha}} \mleft[ \prod_{i = 1}^D (\cosh (2 \alpha_i
  \delta^2))^R  (\cosh (2 \alpha_i \delta^4))^{R'} \mright], 
  \label{eq:boundinnerexpect:partial_improve_attempt}
\end{align}
where \eqref{eq:boundinnerexpect:partial_improve_attempt} follows from the following lemma, whose proof we defer to the end of the section:
\begin{restatable*}[]{lemma}{scaryExpectation}
  \label{lemma:average_of_product_with_binomial_raise_by_m:restated_improve_attempt}
  There
  exists an absolute constant $\delta_0 \approx 0.96$ such that the following
  holds. Let $K \geq 1$ and $R_1, \ldots, R_K$ be integers, and $\delta_1,
  \ldots, \delta_K \in (0, \delta_0]$. Suppose that $\kappa_{j, 1}, \ldots,
  \kappa_{j, D} \sim \tmop{Bin} \mleft( R_j, \frac{1}{2} \mright)$, are i.i.d.,
  and mutually independent across $1 \leq j \leq K$, and $z_j \eqdef \frac{1 +
  \delta_j}{1 - \delta_j}$. Then
  \[ \mathbb{E} \mleft[ \mleft( \frac{1}{D}  \sum_{i = 1}^D \prod_{j = 1}^K
     z_j^{\kappa_{j, i}} \mright)^m \mright] \leq \mleft( \prod_{j = 1}^K
     z_j^{\frac{m}{2} R_j} \mright) \mathbb{E}_{\orgvec{\alpha}} \mleft[
     \prod_{i = 1}^D \prod^K_{j = 1} \cosh (2 \alpha_i \delta_j) \mright], \]
  where $(\alpha_1, \ldots, \alpha_D)$ follows a multinomial distribution with
  parameters $m$ and $(1/D,\dots,1/D)$.
\end{restatable*}
We now focus on the expectation on the right (last factor of the RHS of~\eqref{eq:boundinnerexpect:partial_improve_attempt}): using that $\cosh u \leq \min(e^{u^2/2}, e^u)$ for $u\geq 0$, we have, setting $\Delta\eqdef 1/\delta^2 = n/\eps^2$,
\begin{align}
  \mathbb{E}_{\orgvec{\alpha}} &\mleft[ \prod_{i = 1}^D (\cosh (2 \alpha_i
  \delta^2))^R  (\cosh (2 \alpha_i \delta^4))^{R'} \mright] \nonumber\\
  & \leq \mathbb{E}_{\orgvec{\alpha}} \mleft[ \prod_{i = 1}^D \min(e^{2 \alpha_i \delta^2 R}, e^{2\alpha_i^2 \delta^4 R}) e^{2 \alpha_i^2 \delta^8 R'} \mright] \nonumber\\
  & \leq \mathbb{E}_{\orgvec{\alpha}} \mleft[ \prod_{i = 1}^D e^{2\alpha_i \delta^2 R \mathbbm{1} [\alpha_i > \Delta]} e^{2 \alpha_i^2 \delta^4 R  \mathbbm{1} [\alpha_i \leq \Delta]} e^{2 \alpha_i^2 \delta^8 R'}
  \mright] \nonumber\\
  & \leq \mathbb{E} \mleft[ \prod_{i = 1}^D e^{8 \alpha_i
  \delta^2 R \mathbbm{1} [\alpha_i > \Delta]} \mright]^{1/4} \mathbb{E} \mleft[ \prod_{i =  1}^D e^{8 \alpha_i^2 \delta^4 R \mathbbm{1} [\alpha_i \leq \Delta]} \mright]^{1 / 4}  \mathbb{E} \mleft[ \prod_{i = 1}^D e^{4  \alpha_i^2 \delta^8 R'} \mright]^{1 / 2} \label{eq:three_scary_expactations_multinomial}
\end{align}
where the last step comes from the generalized H\"older inequality (or, equivalently, two applications of the Cauchy--Schwarz
inequality), and the threshold $\Delta$ was chosen as the value for which the term realizing the minimum changes. We first bound the product of the last two expectations:
\begin{align}
  \mathbb{E} \mleft[ \prod_{i = 1}^D e^{8 \alpha_i^2 \delta^4 R
  \mathbbm{1} [\alpha_i\leq \Delta]} \mright]^{1 / 4} &\mathbb{E} \mleft[ \prod_{i = 1}^D e^{4 \alpha_i^2 \delta^8 R'}
  \mright]^{1 / 2} \notag\\
  & \leq \mathbb{E} \mleft[ \prod_{i = 1}^D e^{8 \delta^4 R \min(\alpha_i^2, \Delta^2)} \mright]^{1 / 4} \mathbb{E} \mleft[
  \prod_{i = 1}^D e^{4 \alpha_i^2 \delta^8 R'} \mright]^{1 /
  2} \nonumber\\
  & \leq  (\mathbb{E} [e^{8 \min (\alpha_j^2, \Delta^2) \delta^4 R}])^{D /
  4} (\mathbb{E} [e^{4 \alpha_1^2 \delta^8 R'}])^{D / 2} 
  \label{equation:step_negative_association}\\
  & \leq  \exp \mleft( 32 \delta^4  \frac{m^2}{D} R \mright) \exp \mleft( 32\delta^8 \frac{m^2}{D} R' \mright) .  
  \label{equation:apply_MGF_min_square_bin}
\end{align}
where we applied negative association (see, e.g., \citet[Theorem 13]{Dubhashi96ballsand}) on both expectations for~\eqref{equation:step_negative_association};
and then got~\eqref{equation:apply_MGF_min_square_bin} by Lemmas~\ref{lemma:MGF_Binomial_square_bound_with_min} and~\ref{lemma:MGF_Binomial_square_bound} (for the latter, noting that $t m = 2 \delta^8 m R' \leq 1 / 16$; and, for the former, assuming with little loss of generality that $\eps \leq 1/(4\sqrt{2})$). 
Applying Lemma~\ref{lemma:bound_on_truncated_multinomial_MGF} to the first
(remaining) factor of the LHS above as $8\delta^2 R\leq 4$ and $D \geq \Omega(1)$, we get
\begin{align*}
  &\mathbb{E} \mleft[ \prod_{i = 1}^D e^{8 \alpha_i \delta^2 R
  \mathbbm{1} [\alpha_i > \Delta]} \mright]^{1/4} \mathbb{E} \mleft[ \prod_{i = 1}^D e^{8 \alpha_i^2 \delta^4 R \mathbbm{1} [\alpha_i \leq \Delta]} \mright]^{1/ 4}\mathbb{E} \mleft[ \prod_{i = 1}^D e^{4 \alpha_i^2
  \delta^8 R'} \mright]^{1 / 2}\\
  &\qquad \leq (1 + o (1)) \cdot \exp \mleft( 32 \delta^4  \frac{m^2}{D} R
  \mright) \exp \mleft( 32 \frac{m^2}{D} \delta^8 R' \mright)\\
  &\qquad = (1 + o(1))\exp (32C'^2 R),
\end{align*}
recalling that $R' \leq N/4\leq n/4$, and our assumption that $m \leq C'\sqrt{D}n/\eps^2$. {Combining~\eqref{eq:boundexpect:partial}, \eqref{fact:intermedia_upper_bound_on_multigraph_cycle_term:restated} and~\eqref{eq:boundinnerexpect:partial_improve_attempt}, what we showed is}
 \begin{align*}
  \mathbb{E}_{\theta, \theta'} [(1 + \tau (\theta, \theta'))^m]
  &\leq {(1 + o(1)) e^{128 \frac{m \varepsilon^2}{\sqrt{D} n}}}\mathbb{E}_{\lambda, \lambda'}  \mleft[ (1 - (4 \delta)^2)^{mR}(1 - (4 \delta)^4)^{mR'} z_1^{\frac{mR}{2}} z_2^{\frac{mR'}{2}}e^{32C'^2 \cdot R} \mright]\\
  &= {(1 + o(1)) e^{128 \cdot C'}}\mathbb{E}_{\lambda, \lambda'}  \mleft[ (1 - (4 \delta)^4)^{\frac{mR}{2}}(1 - (4 \delta)^8)^{\frac{mR'}{2}}e^{32C'^2 \cdot R} \mright]\\
  &\leq {(1 + o(1)) e^{128 \cdot C'}}\mathbb{E}_{\lambda, \lambda'}  \mleft[e^{32C'^2 \cdot R} \mright]\,,
\end{align*}
where the equality follows from the definition of $z_1,z_2$. To conclude, we will use the fact that, for every $k\geq 0$,
\begin{equation}
  \label{eq:tail:bound:R}
  \Pr [R > k] \leq \frac{1}{k!},
\end{equation}
which was established in~\citet[p.46]{CanonneDKS20}. By summation by parts, one can show that this implies
\[
    \mathbb{E}_{R}  \mleft[ e^{\alpha R}  \mright] \leq 1+(1-e^{-\alpha}) (e^{e^\alpha}-1)
  \xrightarrow[\alpha\to 0]{} 1+ \alpha(e-1)
\]
for any $\alpha>0$, and so, in our case,
%
 \begin{align}
  \mathbb{E}_{\theta, \theta'} [(1 + \tau (\theta, \theta'))^m]
  &\leq {(1 + o(1)) e^{128 \cdot C'}}\mleft( 1+(1-e^{-32C'^2}) (e^{e^{32C'^2}}-1) \mright)
\end{align}
In particular, the RHS can be made arbitrarily close to $1$ by choosing a small enough value for the constant $C'$ (in the bound for $m$). By~\eqref{eq:Ingster:restated}, this implies the desired bound on $\norm{Q - U^{\otimes m} }_1^2$, and thus establishes Lemma~\ref{theorem:lower_bound_uniformity_testing_on_Bayes_Net}. \qed

\paragraph{The remaining technical lemma.} It only remains to establish Lemma~\ref{lemma:average_of_product_with_binomial_raise_by_m:restated_improve_attempt}, which we do now.
\scaryExpectation
\begin{proof}[Proof
of~Lemma~\ref{lemma:average_of_product_with_binomial_raise_by_m:restated_improve_attempt}]
We will require the following simple fact, which follows from the multinomial theorem and the definition of the multinomial distribution:
\begin{fact}\label{fact:power_of_average_equates_multinomial_expect}Let $D$ be a
positive integer and $m$ be a non-negative integer. For any $x_1, \ldots, x_D
\in \mathbb{R}$, we have
\[ \mleft( \frac{1}{D}  \sum_{i = 1}^D x_i \mright)^m =\mathbb{E}_{\alpha_1,
   \ldots, \alpha_D} \mleft[ \prod_{i = 1}^D x_i^{\alpha_i} \mright], \]
where $(\alpha_1, \ldots, \alpha_D)$ follows a multinomial distribution with
parameters $m$ and $(1/D, \ldots, 1/D)$.
\end{fact}

We now apply Fact~\ref{fact:power_of_average_equates_multinomial_expect} inside
the expectation of the LHS of the statement. Note that the sets of random variables $\orgvec{\alpha} =
\{\alpha_1, \ldots, \alpha_D \}$, $\orgvec{\kappa}_1 = \{\kappa_{1, 1},
\ldots, \kappa_{1, D} \}, \ldots, \orgvec{\kappa}_K = \{\kappa_{K, 1}, \ldots,
\kappa_{K, D} \}$ are mutually independent, since $\orgvec{\alpha}$ are a
set of auxiliary random variables derived from an averaging operation and by
the assumption on $\orgvec{\kappa}_j$; and we have that $\kappa_{j, 1},
\ldots, \kappa_{j, D}$ are i.i.d.,
\begin{align*}
  \mathbb{E} \mleft[ \mleft( \frac{1}{D}  \sum_{i = 1}^D \prod^K_{j = 1}
  z_j^{\kappa_{j, i}} \mright)^m \mright] & =  \mathbb{E} \mleft[
  \mathbb{E}_{\alpha_1, \ldots, \alpha_D} \mleft[ \prod_{i = 1}^D \mleft(
  \prod^K_{j = 1} z_j^{\kappa_{j, i}} \mright)^{\alpha_i} \mright] \mright]\\
  & = \mathbb{E}_{\orgvec{\alpha}, \orgvec{\kappa}_j} \mleft[ \prod_{i =
  1}^D \prod^K_{j = 1} z_j^{\alpha_i \kappa_{j, i}} \mright]\\
  & = \mathbb{E}_{\orgvec{\alpha}} \mleft[ \mathbb{E}_{\orgvec{\kappa}_j}
  \mleft[ \prod_{i = 1}^D \prod^K_{j = 1} z_j^{\alpha_i \kappa_{j, i}} \mright]
  \mright]\\
  & = \mathbb{E}_{\orgvec{\alpha}} \mleft[ \mleft[ \prod_{i = 1}^D \prod^K_{j
  = 1} \mathbb{E}_{\orgvec{\kappa}_j} [z_j^{\alpha_i \kappa_{j, i}}] \mright]
  \mright]\\
  & = \mathbb{E}_{\orgvec{\alpha}} \mleft[ \mleft[ \prod_{i = 1}^D \prod^K_{j
  = 1} \mleft( \frac{1 + z_j^{\alpha_i}}{2} \mright)^{R_j} \mright] \mright] \tag{Probability-Generating Function of a Binomial}\\
  & = \mathbb{E}_{\orgvec{\alpha}} \mleft[ \prod_{i = 1}^D \prod^K_{j = 1}
  z_j^{\frac{\alpha_i R_j}{2}} \mleft( \frac{z_j^{- \alpha_i / 2} +
  z_j^{\alpha_i / 2}}{2} \mright)^{R_j} \mright]\\
  & = \mleft( \prod^K_{j = 1} z_j^{\frac{mR_j}{2}} \mright)
  \mathbb{E}_{\orgvec{\alpha}} \mleft[ \prod_{i = 1}^D \prod^K_{j = 1} \mleft(
  \frac{z_j^{- \alpha_i / 2} + z_j^{\alpha_i / 2}}{2} \mright)^{R_j} \mright] .
\end{align*}
Next, we will simplify the
expression left inside by upper bounding it, using the fact that, given our
assumption on $\delta_j$ being bounded above by $\delta_0$, we have $z_j =
\frac{1 + \delta_j}{1 - \delta_j} \leq e^{4 \delta_j}$. Thus,
\begin{align*}
  \mathbb{E}_{\orgvec{\alpha}} \mleft[ \prod_{i = 1}^D \prod^K_{j = 1} \mleft(
  \frac{z_j^{- \alpha_i / 2} + z_j^{\alpha_i / 2}}{2} \mright)^{R_j} \mright] &
  \leq \mathbb{E}_{\orgvec{\alpha}} \mleft[ \prod_{i = 1}^D \prod^K_{j =
  1} \mleft( \frac{e^{- 2 \alpha_i \delta_j} + e^{2 \alpha_i \delta_j}}{2}
  \mright)^{R_j} \mright]\\
  & = \mathbb{E}_{\orgvec{\alpha}} \mleft[ \prod_{i = 1}^D \prod^K_{j = 1}
  (\cosh (2 \alpha_i \delta_j))^{R_j} \mright]
\end{align*}
as claimed.
\end{proof}

\subsection{Product Distributions Are Far from Mixture of Trees (Lemma~\ref{lemma;distance:mixture:of:trees})}
\label{ssec:mixturetree:farness}
In this subsection, we outline the proof of Lemma \ref{lemma;distance:mixture:of:trees}.
Our argument starts with Lemma~\ref{lemma:technical:TV_lower_bound_of_its_marginals}, which allows us to relate the total variation distance between the mixture and the product of its marginals to a simpler quantity, the difference between two components of this mixture. 
\begin{lemma}
\label{lemma:technical:TV_lower_bound_of_its_marginals}
Let $p$ be a distribution on $\{0, 1\}^N \times \{0, 1\}^M$ (with $N,M\geq 2$), and denote its
marginals on $\{0, 1\}^N$, $\{0, 1\}^M$ by $p_1, p_2$ respectively. Then, if $p_1$ is uniform, %
\[ \begin{array}{l}
    d_{\tmop{TV}} (p, p_1 \otimes p_2)
  \end{array} \geq d_{\tmop{TV}} (p (\cdummy \mid x_1 = 0), p (\cdummy  \mid x_1
  = 1))\,. \]
\end{lemma}   
This in turn will be much more manageable, as the parameters of these two mixture components are independent, and thus analyzing this distance can be done by analyzing Binomial-like expressions. This second step is reminiscent of \cite[Lemma 8]{CanonneDKS20}, which can be seen as a simpler version involving only one Binomial instead of two:
\new{
  \begin{restatable}[]{lemma}{DistanceLbBinomialsConditionalTree}
    \label{lemma:technical_distance_lb_2xbinomials_conditional_tree} 
    There exist $C_1,C_2>0$ such that the following holds. Let $\eps\in(0,1]$ and $n \geq C_1$, and let $a,b$ be two integers such that $a+b=n$  and $b \geq
    \frac{1}{4} n$. Then, for $\delta \eqdef \frac{\eps}{\sqrt{n}}$, we have
    \[ 
      \frac{(1 - \delta)^n}{2^n}  \sum_{k_1 = 0}^{a}\sum_{k_2 = 0}^{b} \binom{a}{k_1} \binom{b}{k_2} \mleft|
        \mleft( \mleft( \frac{1 + \delta}{1 - \delta} \mright)^{k_1 + k_2} - \mleft(
        \frac{1 + \delta}{1 - \delta} \mright)^{k_1 + b - k_2} \mright) \mright|
        \geq C_2\eps . \]
  \end{restatable}
}
This parameter $b$ corresponds to the
difference between the orientations parameters 
$\mu, \mu'$ being large, which happens with high constant probability as long
as $n$ is large enough. The proof of Lemma
\ref{lemma:technical_distance_lb_2xbinomials_conditional_tree} is deferred to Appendix
\ref{proof:lemma:technical_distance_lb_2xbinomials_conditional_tree}, and we hereafter proceed with the rest of the argument.
For fixed $\theta$ and $x_2, \ldots, x_d$, $z \assign \frac{1 + 4 \delta}{1 -
4 \delta}$.  We will denote by $\mu, \mu'$ the two (randomly chosen) orientation parameters corresponding to the mixture components indexed by $(0,x_2,\dots, x_d)$ and $(1,x_2,\dots, x_d)$.     
By
Lemma \ref{lemma:technical:TV_lower_bound_of_its_marginals} and Lemma
\ref{lemma:technical_approximated_lb_product_distributions}, for any product
distribution $q$,
\begin{align}
  2 d_{\tmop{TV}} (p, q) & \geq \frac{1}{3\cdot 2^N}  (1 - 4 \delta)^{\frac{N}{2}}  \sum_{x_{d + 1}, \ldots,
  x_n} |z^{c (\lambda, \mu, x)} - z^{c (\lambda, \mu', x)} | \label{eq:expression:tv}
\end{align}
Let $S_1$ denote the set of pairs in the child nodes with common parameters
between $\mu$ and $\mu'$, and $S_2$ the set of pairs with different
parameters (that is, the definition of $S_1, S_2$ is essentially that of $A$ and $B$ from the previous section (p.\pageref{def:AB:lb}), but for equal matching parameters $\lambda = \lambda'$). In particular, we have that $|S_2 | = \tmop{Hamming} (\mu, \mu') \sim
\tmop{Binomial} \left( \frac{N}{2}, \frac{1}{2} \right)$ and $|S_1 \cup S_2 |
= \frac{N}{2}$. Let $\tilde{c}(S, \mu, x)$ be the analogue of $c_{\tmop{ch}}
(\lambda, \mu, x)$ from~\eqref{def:ch}, but only on a subset of pairs $S$ instead of $\{d,\dots,n\}^2$; i.e.,
\[ \tilde{c}(S, \mu, x) \eqdef \left| \left\{ (i, j) \in S : \exists k \in
    \mathbb{N}, \lambda_k = (i - d + 1, j - d + 1) \infixand (- 1)^{x_i + x_j} = (-
    1)^{\mu_k} \right\} \right| . \]
Given any $x, \mu \infixand \mu'$, the following holds from the definitions of $\tilde{c}$ and $c_{\tmop{ch}}$:
\begin{itemize}
  \item Since $S_1 \cup S_2$ contains all the pairs, $c_{\tmop{ch}} (\lambda, \mu, x)
  = \tilde{c}(S_1, \mu, x) + \tilde{c}(S_2, \mu, x)$ (similarly for $\mu'$).
  
  \item Since $S_1$ (resp., $S_2$) contains exactly the pairs whose orientation is the same (resp., differs) between $\mu$ and $\mu'$, we have
  $\tilde{c}(S_1, \mu', x) = \tilde{c}(S_1, \mu, x)$
  and
  $\tilde{c}(S_2, \mu', x) + \tilde{c}(S_2, \mu, x) = |S_2 |$
  
  \item For a fixed matching and a partition $S_1,S_2$ of its $N/2$ pairs, given an orientation vector $\mu\in\{0,1\}^{N/2}$, and fixed values $0\leq k_1\leq |S_1|$, $0\leq k_2\leq |S_2|$, there are $2^{N/2}\binom{|S_1|}{k_1}\binom{|S_2|}{k_2}$ different vectors $x\in\{0,1\}^N$ such that $\tilde{c}(S_1, \mu, x)=k_1$ and $\tilde{c}(S_2, \mu, x)=k_2$.
\end{itemize}
Using these properties, we have, assuming $|S_2 | \geqslant \frac{1}{4} \cdot \frac{N}{2}$ and $N$ bigger than some constant,
\begin{align*}
    \sum_{x \in \{0, 1\}^N}  |z^{c_{\tmop{ch}} (\lambda,
  \mu, x)} - z^{c_{\tmop{ch}} (\lambda, \mu', x)} |
  &= \sum_{x \in \{0, 1\}^N} |z^{\tilde{c}(S_1, \mu, x) + \tilde{c}(S_2, \mu, x)} - z^{\tilde{c}(S_1, \mu', x) + \tilde{c}(S_2, \mu', x)} |\\
  &= \sum_{x \in \{0, 1\}^N} |z^{\tilde{c}(S_1, \mu, x)+\tilde{c}(S_2, \mu, x)} - z^{\tilde{c}(S_1, \mu, x)+|S_2|-\tilde{c}(S_2, \mu, x)} |\\
  &= 2^{\frac{N}{2}}\sum_{k_1=0}^{|S_1|} \binom{|S_1|}{k_1}\sum_{k_2=0}^{|S_2|} \binom{|S_2|}{k_2} |z^{k_1+k_2} - z^{k_1+|S_2|-k_2} |\\&\geq C\cdot \frac{2^N}{(1-4\delta)^{N/2}}\cdot \eps,
\end{align*}
where $C>0$ is an absolute constant, and for the last inequality we invoked Lemma~\ref{lemma:technical_distance_lb_2xbinomials_conditional_tree}.
 Recalling now that $|S_2 | \sim \tmop{Bin}\mleft( \frac{N}{2},
    \frac{1}{2} \mright)$, for $N$ large enough we also have
\[ \Pr \mleft[ |S_2| \geq \frac{N}{8} \mright] \geq 1 - e^{
    - \frac{N}{16}} > 9/10. \]
Thus, combining the two along with~\eqref{eq:expression:tv}, we conclude that
\[
\Pr [d_{\tmop{TV}} (P_{\theta} (\cdot \mid x_1 = 0, x_2, \ldots, x_d),
    P_{\theta} (\cdot \mid x_1 = 1, x_2, \ldots, x_d)) \geq \Omega
    (\varepsilon)] \geq 9/10,
\]
establishing Lemma~\ref{lemma;distance:mixture:of:trees}. \hfill$\blacksquare$
%

%% file: sec-mgfs.tex
In this section, we establish various self-contained results on the moment-generating functions (MGF) and stochastic dominance of Binomials, truncated (or ``capped'') Binomials, and multinomial distributions, which we used extensively in Section~\ref{ssec:distinguish:uniform} and should be of independent interest.
  Notably, derivations from Section~\ref{sec:lb:outline} following~\eqref{eq:three_scary_expactations_multinomial} are direct consequences of the three lemmas in the section: Lemma \ref{lemma:MGF_Binomial_square_bound}, Lemma \ref{lemma:bound_on_truncated_multinomial_MGF} and Lemma \ref{lemma:MGF_Binomial_square_bound_with_min} below, which we restate and establish later in this section.

  \begin{restatable*}{lemma}{MGFSquaredBinomial}
    \label{lemma:MGF_Binomial_square_bound}
    Let $X \sim \tmop{Bin} (m, p)$. Then, 
    for any $t$ such that $0 < tm \leq 1/16$,
    \[
    \mathbb{E} [e^{tX^2}] \leq \exp (16tm^2 p^2 + 2tmp).
    \]
  \end{restatable*}

  \begin{restatable*}[]{lemma}{MGFCappedMultinomial}
    \label{lemma:bound_on_truncated_multinomial_MGF}
    Suppose $\orgvec{\alpha} =
    (\alpha_1, \ldots, \alpha_D)^T$ follows a multinomial distribution with
    parameters $m$ and $(1/D, \ldots, 1/D)$, and $\Delta \geq 40 D \gg c^4$ be such that $m {\leq} c \sqrt{D} \Delta$. Then, for any $t \leq 4$ and $D \geq \Omega(1)$, we have
      \[ 
    \mathbb{E} \mleft[ \prod_{i = 1}^D e^{t  \alpha_i
      \mathbbm{1} [\alpha_i > \Delta]} \mright] \leq 1 + c \sqrt{D} \exp \mleft( - \frac{1}{80} \Delta \log D \mright) . 
      \]
  \end{restatable*}

  \begin{restatable*}[]{lemma}{MGFSquaredCappedBinomial}
    \label{lemma:MGF_Binomial_square_bound_with_min}
    Let $X' \sim \tmop{Bin} (m, p)$, and $X \eqdef \min (X', \Delta)$, for some $\Delta \leq m$. Then,
    for any $t$ such that $0 < t\Delta \leq 1 / 8$ and $0 < t m p \leq 1 /
    16$, we have
    \[ \mathbb{E} [e^{tX^2}] \leq \exp (16 tm^2 p^2 + 2 t m p) . \]
  \end{restatable*}
\subsection{Bounds on moment-generating functions}
We start with some relatively simple statements:
\begin{fact}
\label{fact:bound_on_MGF_Binomial_not_general}If $X \sim \tmop{Bin}
(m, p)$, then, for any $0 \leq t \leq 1$,
$\mathbb{E} [e^{t X}] \leq \exp (2 t mp).$
\end{fact}
\begin{proof}
  This follows from computing explicitly 
  $
  \mathbb{E} [e^{t X}] = (1 + p (e^{t} - 1))^m \leq (1 + 2
    t p)^m \leq e^{2 t mp},
    $
    where the first inequality uses that $t\leq 1$. 
\end{proof}
\noindent We will also require the following decoupling inequality:
\begin{lemma}
  \label{lemma:decoupling}Let $F\colon \mathbb{R} \to \mathbb{R}$ be a convex,
  {\tmem{non-decreasing}} function, and $X = (X_1, \ldots, X_n)$ be a vector
  of independent {\tmem{non-negative}} random variables. Then
  \[ \E \mleft[ F \mleft( \sum_{i \neq j} X_i X_j \mright) \mright] \leq \E \mleft[
    F \mleft( 4 \sum_{i, j} X_i Y_j \mright) \mright] \]
  where $Y$ is an independent copy of $X$.
\end{lemma}
\begin{proof}
  \label{proof:lemma:decoupling}
  Introduce a vector of independent
  (and independent of $X$) $\tmop{Bern} (1 / 2)$ random variables $\delta =
  (\delta_1, \ldots, \delta_n)$; so that $\mathbb{E} [\delta_i (1 - \delta_j)]
  = \frac{1}{4} \mathbbm{1}_{i \neq j}$. For any realization of $X$, we can write
  \[ \sum_{i \neq j} X_i X_j = 4 \E_{\delta} \mleft[ \sum_{i \neq j} \delta_i
     (1 - \delta_j) X_i X_j \mright] = 4 \E_{\delta} \mleft[ \sum_{i, j}
     \delta_i (1 - \delta_j) X_i X_j \mright], \]
  and so, by Jensen's inequality and Fubini, as well as independence of $X$ and $\delta$,  
  \begin{align*}
    \mathbb{E}_X \mleft[ F \mleft( \sum_{i \neq j} X_i X_j \mright) \mright] &
    =\mathbb{E}_X \mleft[ F \mleft( 4 \E_{\delta} \mleft[ \sum_{i, j} \delta_i (1
    - \delta_j) X_i X_j \mright] \mright) \mright]\\
    & \leq \mathbb{E}_{\delta} \mleft[ \mathbb{E}_X \mleft[ F \mleft( 4
    \sum_{i, j} \delta_i (1 - \delta_j) X_i X_j \mright) \mright] \mright]
  \end{align*}
  
  This implies that there exists some realization $\delta^{\ast} \in \{0,
  1\}^n$ such that
  \[ \mathbb{E}_X \mleft[ F \mleft( \sum_{i \neq j} X_i X_j \mright) \mright]
     \leq \mathbb{E}_X \mleft[ F \mleft( 4 \sum_{i, j} \delta^{\ast}_i (1 -
     \delta^{\ast}_j) X_i X_j \mright) \mright] \hspace{0.17em} . \]
  Let $I \assign \{i \in [n] : \delta^{\ast}_i = 1\}$. Then $\sum_{i, j}
  \delta^{\ast}_i  (1 - \delta^{\ast}_j) X_i X_j = \sum_{(i, j) \in I \times
  I^c} X_i X_j$, and we get
  
  \begin{align}
    \mathbb{E}_X \mleft[ F \mleft( \sum_{i \neq j} X_i X_j \mright) \mright] &
    \leq \mathbb{E}_X \mleft[ F \mleft( 4 \sum_{(i, j) \in I \times I^c}
    X_i X_j \mright) \mright] \label{eq:decoupling:intermediate}\\
    & =\mathbb{E}_X \mleft[ F \mleft( 4 \sum_{(i, j) \in I \times I^c} X_i Y_j \mright) \mright] \notag\\
    & \leq \mathbb{E}_X \mleft[ F \mleft( 4 \sum_{(i, j) \in I \times I^c} X_i Y_j + 4 \sum_{(i, j) \nin I \times I^c} X_i Y_j \mright) \mright] \notag\\
    & =\mathbb{E}_X \mleft[ F \mleft( 4 \sum_{i, j} X_i Y_j \mright) \mright] \notag
\;,
  \end{align}
  where the equality uses the fact that $(X_i)_{i \in I}$ and $(X_j)_{j \in
  I^c}$ are independent (as $I, I^c$ are disjoint), and so replacing $\sum_{j
  \in I^c} X_j$ by the identically distributed $\sum_{j \in I^c} Y_j$ does not
  change the expectation; and the second inequality uses monotonicity of $F$
  and non-negativity of $X, Y$, as $4 \sum_{(i, j) \nin I \times I^c} X_i Y_j
  \geq 0$. (Note that up to (and including)~\eqref{eq:decoupling:intermediate}, the assumption that the $X_i$'s are independent is not necessary; we will use this fact later on.)
\end{proof}
Note that compared to the usual version of the inequality, we do not require
that the $X_i$'s have mean zero; but instead require that they be
non-negative, and that $F$ be monotone. We will, in the next lemma, apply
Lemma~\ref{lemma:decoupling} to the function $F (x) = e^{2 tx}$, for some
fixed {\tmem{positive}} parameter $t > 0$ (so that $F$ is indeed
non-decreasing), and to $X_1, \ldots, X_n$ independent Bernoulli r.v.'s. Specifically, we obtain the following bound on the MGF of the square of a Binomial:

\MGFSquaredBinomial
\begin{proof}
  Write $X=\sum_{i = 1}^m X_i$, where the $X_i$ are i.i.d.\ $\tmop{Bern} (p)$ (in particular, $X_i = X_i^2$).         
  Then, by the Cauchy--Schwarz inequality and the decoupling inequality from Lemma~\ref{lemma:decoupling}, we have, for $t>0$,
  \begin{eqnarray}
    \mathbb{E} [e^{tX^2}]  
    & = & \mathbb{E} \mleft[ e^{t \sum_i X_i} e^{t \sum_{i \neq j} X_i
    X_j} \mright] \nonumber\\
    & \leq & \sqrt{\mathbb{E} \mleft[ e^{2 t \sum_i X_i} \mright]} 
    \sqrt{\mathbb{E} \mleft[ e^{2 t \sum_{i \neq j} X_i X_j} \mright]} \nonumber\\
    \mleft( \text{decoupling} \mright) & \leq & \sqrt{\mathbb{E} \mleft[ e^{2
    t \sum_i X_i} \mright]}  \sqrt{\mathbb{E} \mleft[ e^{8 t \sum_{i, j} X_i Y_j}
    \mright]} \text{.}  \label{eq:decoupling_step}
  \end{eqnarray}
  where $Y_j \sim \tmop{Bern} (p)$ are i.i.d., and independent of the $X_i$'s. Let $Y = \sum_{i = 1}^m Y_i\sim \tmop{Bin}(m, p)$.         
  From Fact~\ref{fact:bound_on_MGF_Binomial_not_general}, as long as $2 t
  \leq 1$, $8 t m \leq 1$, and $16 tmp \leq 1$
  (all conditions satisfied in view of our assumption),
  \[ \mathbb{E}_{X, Y} [e^{8 tXY}] =\mathbb{E}_X [\mathbb{E}_Y [e^{8 tXY}]]
    \leq \mathbb{E}_X [e^{16 tXmp}] \leq e^{32 tm^2 p^2},
  \]
  and $\mathbb{E} \mleft[ e^{2t X} \mright] \leq e^{4tmp}$. Going back to~\eqref{eq:decoupling_step}, this implies
  \[
    \mathbb{E} [e^{tX^2}] 
    \leq \sqrt{\exp \left(4tmp\right)}\sqrt{\exp \left(32 tm^2 p^2 \right)}
    = \exp \left(2 tmp + 16 tm^2 p^2 \right)\,,
  \]
  concluding the proof.
\end{proof}
We will prove an MGF bound on the truncated Multinomial in Lemma \ref{lemma:bound_on_truncated_multinomial_MGF} (noting that using MGF bound of Multinomial distribution is not nearly enough), as required by our analysis on the independence testing lower bound; prior to that, we will need two important lemmas: Lemma \ref{lemma:factorial_fraction_inequality} and Lemma \ref{lemma:multinomial_truncated_MGF}. These two lemmas both try to bound the expression with a uniform and more manageable term.

%
%
%
%
%
\begin{lemma}
  \label{lemma:factorial_fraction_inequality}
  Fix $m,\Delta,D$ such that $\frac{m}{\Delta} {\leq}c\sqrt{D}$ for some $c>0$ (and $D> \max(16c^4, e^{100})$). Fix any integer $k > 0$ and a tuple of non-negative integers $(a_1,\dots,a_D)$ summing to $m$ such that $L \eqdef \sum_{i = 1}^k a_i > k\Delta$ (in particular, $k \leq c\sqrt{D}$). 
  Suppose $(\alpha_1, \ldots, \alpha_D)$ follows
  a multinomial distribution with
  parameters $m$ and $(1/D, \ldots, 1/D)$. Then,
  \[
    e^{4 L} \Pr \mleft[ \orgvec{\alpha} = (a_1, \ldots, a_D)
     \mright] \leq m\cdot  \exp (- \frac{1}{5} L \log D) .
    \]
\end{lemma}

\begin{proof}
  Via a multinomial distribution grouping argument, the probability can be
  bounded by considering a grouping of two random variables, $L_1 = \sum_{i = 1}^k
  \alpha_i$ and $L_2 = \sum_{i = k + 1}^D \alpha_i$, where $(L_1,L_2)$ follows
  a multinomial distribution with parameters $m$ and $\mleft( \frac{k}{D},
  \frac{D - k}{D} \mright)$, namely, recalling $L = \sum_{i = 1}^k a_i$ and setting $T \eqdef \sum_{i = k + 1}^D a_i$, 
  \[ \Pr \mleft[ \orgvec{\alpha} = (a_1, \ldots, a_D) \mright] \leq 
  \Pr\mleft[ L_1 = L, L_2 = T    \mright] = \frac{m!}{L! T!}  \mleft( \frac{k}{D}
     \mright)^{L} \mleft( \frac{D - k}{D} \mright)^{T} \]
  Moreover, note that $m=L+T$. Via Stirling's approximation, we have
  \begin{equation}
    \label{equation:stirling_gibbs_application} 
    \frac{m!}{L ! T!}
      \leq \exp(m \log m + \log m - L \log L - T \log T)
  \end{equation}
 from which we can write, taking the logarithm for convenience,
  \begin{align}
    \log\mleft( \frac{m!}{L!T!}  \mleft( \frac{k}{D} \mright)^{L}
    \mleft( \frac{D - k}{D} \mright)^{T} \mright) 
    & \leq \log m - \mleft( L \log \frac{L D}{m k} + T \log \frac{T  D}{m (D - k)}
    \mright)
    \label{equation:everything-grouped-together-log-sum}\\
    & = \log m - \mleft( L \log \frac{LD}{m
    k }   + T \log \mleft( \frac{TD}{m ({D^{3 / 4}} -
    k)}  \frac{{D^{3 / 4}} - k}{D - k} \mright) \mright) \label{equation:swapping-D-with-lower-order-D}\\
    & \leq \log m - m \log (D^{1 / 4}) + (m - L) \log \mleft(
    \frac{D - k}{D^{3 / 4} - k} \mright)
    \label{equation:Gibbs-ineqaulity-again}\\
    & = \log m + m \log \mleft( 1 + k\frac{D^{1 / 4} - 1}{D
    - k D^{1 / 4}} \mright) - L \log \mleft( 1 + \frac{D^{1 / 4} - 1}{1 - k /
    D^{3 / 4}} \mright)  \nonumber\\
    & \leq \log m + m k \frac{D^{1 / 4} - 1}{D - kD^{1 / 4}} - L
    \log (D^{1 / 4})
    \label{equation:setting-k-to-extreme}\\
    & \leq  \log m +
    \frac{m k}{D^{1 / 2}}  \frac{1 - 1 / D^{1 / 4}}{D^{1 / 4} - c} - \frac{1}{4} L \log D  \nonumber\\
    & \leq  \log m +  
    \frac{c L}{D^{1 / 4} - c} - \frac{1}{4} L \log D 
    \label{equation:upperbound-for-mk-over-D-square}\\
    & \leq  \log m + L - \frac{1}{4} L \log D  \tag{as $c/D^{1/4}\leq 1/2$} 
  \end{align}
  where we used Gibbs' inequality for \eqref{equation:Gibbs-ineqaulity-again};
  we then have (\ref{equation:setting-k-to-extreme}) by $\log (1 + x)
  \leq x$ for the first term. \eqref{equation:upperbound-for-mk-over-D-square} then follows from $k \leq c\sqrt{D}$ and $km {\leq} k\Delta \cdot c\sqrt{D} \leq L\cdot  c\sqrt{D} $. Finally,
  \[
    e^{4 L}  \frac{m!}{L!T!}  \mleft( \frac{k}{D}
    \mright)^{L} \mleft( \frac{D - k}{D} \mright)^{T} 
    \leq   \exp (5L- \frac{1}{4}L \log D)
     \leq  m \exp \mleft( - \frac{1}{5} L \log D \mright) \,,
  \]
  the last inequality as long as $\log D > 100$.
\end{proof}

\begin{lemma}
  \label{lemma:multinomial_truncated_MGF}
  Suppose $\orgvec{\alpha} =
  (\alpha_1, \ldots, \alpha_D)^T$ follows a multinomial distribution with
  parameters $m$ and $(1/D, \ldots, 1/D)$, and that $\frac{m}{\Delta} {\leq} c\sqrt{D}$ for some $c>0,\Delta\geq 1$ with $\Delta \geq 40D \geq \Omega(c^4)$ and $D > \Omega(1)$. For any integer $c \sqrt{D} \geq k \geq 1$ and any $t\leq 4$,
    \[ 
  \mathbb{E} \mleft[ \prod_{i : \alpha_i > \Delta} e^{t \alpha_i}
     \cdot \mathbbm{1} [\nu(\orgvec{\alpha})= k] \mright] \leq \exp \mleft( -
     \frac{1}{80} \Delta \log (D) \mright) . 
    \]
    where $\nu(\orgvec{\alpha}) \eqdef |\{ i: \alpha_i > \Delta \}|$ denotes the number of coordinates of $\orgvec{\alpha}$ greater than $\Delta$.
\end{lemma}
\begin{proof}
  Without loss of generality, (as later, we will sum over all combinations)
  assume that $\alpha_1, \ldots, \alpha_k$ are the coordinates larger than $\Delta$, for some integer $k$; and denote their sum by $L$. Note that we then have $k\Delta < L \leq m {\leq} c\Delta\sqrt{D}$, and thus $0 \leq k \leq c\sqrt{D}$.
  \begin{equation}
    \mathbb{E} \mleft[ \prod_{i : \alpha_i > \Delta} e^{4\alpha_i}
    \cdot \mathbbm{1} [\nu(\orgvec{\alpha})= k] \mright]
    = \binom{D}{k} \!\!\sum_{\alpha_1, \ldots, \alpha_k > \Delta} \sum_{\alpha_{k
    + 1}, \ldots, \alpha_D \leq \Delta} \!\!\!e^{4 \sum_{i = 1}^k \alpha_i} \Pr
    [\orgvec{\alpha} = (\alpha_1, \ldots, \alpha_D)] 
    \label{equation:MGF_bound_of_truncated_multinomial}
  \end{equation}
  A uniform bound on any $\alpha_1, \ldots, \alpha_D$ as specified can be
  obtained from Lemma \ref{lemma:factorial_fraction_inequality}; and, combining it
  with (\ref{equation:MGF_bound_of_truncated_multinomial}), we have an
  expression that does not depend on the value of $\orgvec{\alpha}$; from which\footnote{We have the number of terms in the summation upper bounded by the following
  analysis: $(m - \Delta)^k$ is an upper bound of combinations of $\alpha_1,
  \ldots, \alpha_k$ with values larger than $\Delta$; and similarly, $(\Delta + 1)^{D -
  k}$ will be the upper bound for the combinations of $\alpha_{k + 1}, \ldots,
  \alpha_D$ with values up to $\Delta$.}
  \begin{eqnarray}
    \mathbb{E} \mleft[ e^{4\sum_{i = 1}^k \alpha_i} 
    \mathbbm{1} [\nu(\orgvec{\alpha})= k] \mright] 
    & \leq & \binom{D}{k} \sum_{\alpha_1, \ldots, \alpha_k > \Delta, \alpha_{k + 1},
    \ldots, \alpha_D \leq \Delta} m e^{- \frac{1}{5} \Delta \log D} \nonumber\\
    & \leq & \binom{D}{k} (m - \Delta)^k \Delta^{D - k} \exp \mleft( \log m - \frac{1}{5}
    \Delta \log D \mright) \nonumber\\
    & \leq & \exp \mleft( k \log D + k \log (m - \Delta) + (D - k) \log \Delta + \log m
    - \frac{\Delta}{5} \log (D) \mright) \nonumber\\
    & = & \exp \mleft( k \log (D \cdot \frac{m - \Delta}{\Delta}) + D \log \Delta + \log m -
    \frac{\Delta}{5} \log (D) \mright) \nonumber\\
    & \leq & \exp \mleft( - \frac{1}{5} \Delta \log D+ (D+1) \log \Delta + \log (c\sqrt{D})+
    \frac{3}{2} k \log (c D) \mright) \nonumber\\
    & \leq & \exp \mleft( - \frac{1}{10} \Delta \log D + 2c \sqrt{D}
    \log (c D)  \mright)  \label{weight_log_swap}\\
    & \leq & \exp \mleft( - \frac{1}{80} \Delta \log D \mright) . \nonumber
  \end{eqnarray}
  where (\ref{weight_log_swap}) follows from $20\frac{D}{\log D} \leq \frac{\Delta}{\log \Delta}$, which holds for $\Delta \geq 40D$ and $D$ large enough (larger than some absolute constant); and the last inequality holds, given the above constraints, for $D \geq 16 c^4$.
\end{proof}

\MGFCappedMultinomial
\begin{proof}
  Let $\nu(\orgvec{\alpha}) \eqdef |\{ i: \alpha_i > \Delta \}|$ denote the number of coordinates of $\orgvec{\alpha}$ greater than $\Delta$. Note that $\nu(\orgvec{\alpha}) < L\eqdef \frac{m}{\Delta}$, and that $L = c \sqrt{D}$s by assumption. We break down the expectation by enumerating over the possible values for $\nu(\orgvec{\alpha})$,
  from $0\leq k\leq L$:
  \begin{eqnarray}
    \mathbb{E} \left[ \prod_{i = 1}^D e^{t \alpha_i \mathbbm{1} [\alpha_i > \Delta]} \right] 
    & = & \mathbb{E} \mleft[ \sum_{k = 1}^L \prod_{i : \alpha_i > \Delta} e^{t
    \alpha_i} \cdot \mathbbm{1} [\nu(\orgvec{\alpha}) = k] +
    \mathbbm{1} [\nu(\orgvec{\alpha}) = 0] \mright] \nonumber\\
    & = & \sum_{k = 1}^L \mathbb{E} \mleft[ \prod_{i : \alpha_i > \Delta} e^{t
    \alpha_i} \cdot \mathbbm{1} [\nu(\orgvec{\alpha}) = k] \mright]
    + 1 \cdot \Pr [\nu(\orgvec{\alpha}) = 0] \nonumber\\
    & \leq & L \exp \mleft( - \frac{1}{80} \Delta \log D \mright) + \Pr [\nu(\orgvec{\alpha}) = 0]  
    \label{equation:application_multinomial_truncated_MGF}\\
    & \leq & c D^{1/2} \exp \mleft( - \frac{1}{80} \Delta \log D \mright) + 1\,, \nonumber
  \end{eqnarray}
  where (\ref{equation:application_multinomial_truncated_MGF}) follows from
  Lemma \ref{lemma:multinomial_truncated_MGF}.
\end{proof}

We now state and prove our last lemma, Lemma~\ref{lemma:MGF_Binomial_square_bound_with_min}, on the MGF of the square of a truncated Binomial:

\MGFSquaredCappedBinomial
\begin{proof}
  We will analyze the sampling process in Definition \ref{definition:coupling-decoupling-sampling-procedure-attempt-2}:
  \begin{definition}
    \label{definition:coupling-decoupling-sampling-procedure-attempt-2}
    Fix integers $m \geq \Delta \geq 1$, and let $X'_1,\dots, X'_m$ be i.i.d.\ $\tmop{Bern} (p)$ random variables. Define the distribution of $X_1,\dots,X_m$ through the following sampling process:
    \begin{enumerate}
      \item Initialize $X_i = 0$ for all $i \in [m]$; sample $\{ X_i' \}_{1\leq i\leq m}$ as $m$
      i.i.d.\ $\tmop{Bern} (p)$;
      \item If $\sum_{i \in [m]} X_i' < \Delta$, let $X_i = X_i'$ for all $i \in
      [m]$;
      \item If $\sum_{i \in [m]} X_i' \geq \Delta$, let $\mathcal{S}' = \{ i \in
      [m] \of X_i' = 1 \}$ and let $\mathcal{S}$ be a uniformly random subset of
      $\mathcal{S}'$ of size $\Delta$; set $X_i = X_i'$ for $i \in \mathcal{S}$.
    \end{enumerate}
  \end{definition}
  \noindent Consider a sequence of random variable $X_1,\dots,X_m$ as defined in
  Definition~\ref{definition:coupling-decoupling-sampling-procedure-attempt-2}; each
  $X_i$ (for $1\leq i\leq m$) is supported on $\{ 0, 1 \}$ (so that, in particular, $X_i^2=X_i$); and $X = \sum_{i \in [m]}
  X_i$. By the Cauchy--Schwarz inequality,
  \begin{eqnarray}
    \mathbb{E} [e^{t X^2}] & = & \mathbb{E} \left[ e^{t \sum_{i = 1}^m X_i + t
    \sum_{i \neq j} X_i X_j} \right] \nonumber\\
    & \leqslant & \sqrt{\mathbb{E} \left[ e^{2 t \sum_{i = 1}^m X_i} \right]}
    \sqrt{\mathbb{E} \left[ e^{2 t \sum_{i \neq j} X_i X_j} \right]}
    \nonumber\\
    & \leqslant & \sqrt{\mathbb{E} \left[ e^{2 t \sum_{i = 1}^m X_i} \right]}
    \sqrt{\mathbb{E} \left[ e^{8 t \sum_{(i, j) \in I \times I^c} X_i X_j}
    \right]} 
    \label{equation:coupling-decoupling-application-square-truncated-binomials}\\
    & \leqslant & \sqrt{\mathbb{E} \left[ e^{2 t X} \right]}
    \sqrt{\mathbb{E} [e^{8 t Y_1 Y_2}]} 
    \label{equation:dominance_application}
  \end{eqnarray}
  where $Y_1 \sim \min (\tmop{Bin} (| I |, p), \Delta)$, $Y_2 \sim \min (\tmop{Bin}
  (| I^c |, p), \Delta)$ and $Y_1$ is independent of $Y_2$ (and $(I, I^c)$ is some fixed, but   unknown partition of $[m]$).
  \eqref{equation:coupling-decoupling-application-square-truncated-binomials} follows from the intermediate step~\eqref{eq:decoupling:intermediate} in the proof of Lemma~\ref{lemma:decoupling} (observing that $x\mapsto e^{t x}$ is convex, and
  non-decreasing as $t>0$; and using the remark from that proof about the independence of $X_i$'s not being required up to that step) and (\ref{equation:dominance_application}) follows from Lemma
  \ref{lemma:coupling-for-decoupling-attempt-2}. We will implicitly use Facts~\ref{fact:useful_dominance_order},
  \ref{fact:min_preserves_dominance_order}, and~\ref{fact:dominance_between_binomials} for the remaining calculations, eventually replacing most expressions with $X' \sim \tmop{Bin} (m, p)$.
  
  Recalling that $X \leq X'$ by definition, the first term 
  of~\eqref{equation:dominance_application} can be bounded as $\mathbb{E} [e^{2 t X}] \leqslant e^{4 t m p}$. Moreover, from our assumption, $t Y_1 \leq t \Delta \leq 1 / 8$ and $t m p
  \leq 1 / 16$. Combined with the fact that $Y_1, Y_2$ is dominated by $X
  \sim \min (\tmop{Bin} (m, p), \Delta)$ and thus by $X' \sim \tmop{Bin} (m, p)$,
  we have
    \[
        \mathbb{E} [e^{8 t Y_1 Y_2}] = \mathbb{E}_{Y_1} [\mathbb{E}_{Y_2}
        [e^{8 t Y_1 Y_2}]]
        \leqslant \mathbb{E}_{Y_1} [e^{16 t Y_1 m p}]
        \leqslant e^{32 t m^2 p^2} .
    \]
  Going back to~\eqref{equation:dominance_application},
  this implies
    \[
        \mathbb{E} [\exp \left(t X^2 \right)] \leq \sqrt{\exp \left( 4 t m p \right)}  \sqrt{\exp \left( 32 t m^2 p^2 \right)} = \exp \left(2 t m p + 16 t m^2 p^2 \right),
    \]
  concluding the proof.
\end{proof}

\subsection{Stochastic dominance results between truncated Binomials} %
\begin{fact}\label{fact:useful_dominance_order}
Let $X \sim \tmop{Bin} (m, p)$, and $0 < n \leq m$. Defining $Y \eqdef \min (X, n)$ and $Z \eqdef X \mid X\leq
n$, we have, for every $k\geq 0$,
\[ \Pr [X \geq k] \geq \Pr [Y \geq k] \geq
   \Pr [Z \geq k], \]
i.e., $X \succeq Y \succeq Z$, where $\succeq$ denotes first-order stochastic dominance.
\end{fact}

\begin{proof}
  We can write the PMF of $Z$ and $Y$, for all $0\leq k\leq n$,
  \[ \Pr [Y = k] = \left\{\begin{array}{ll}
       \Pr [X = k], & k < n\\
       \Pr [X \geq n], & k = n
     \end{array}\right., \qquad \Pr [Z = k] = \frac{\Pr [X = k]}{\Pr [X \leq n]} . \]
  It follows that $\Pr [Y \geq k] = \Pr [X \geq k]\mathbbm{1}\{k\leq n\}$, which gives the first part of the statement.
  
  The second part follows from a direct comparison between the two CDF of $Z, Y$: indeed, for $0\leq k \leq n$,
  \begin{eqnarray*}
   \Pr [Y \geq k] \geq \Pr [Z \geq k]
    & \Leftrightarrow & \Pr [X \geq k] \geq \frac{\Pr [n \geq
    X \geq k]}{\Pr [X \leq n]}\\
    & \Leftrightarrow & \Pr [X \geq k] (1 - \Pr [X > n])
    \geq \Pr [X \geq k] - \Pr [X > n]\\
    & \Leftrightarrow & \Pr [X \geq k] \Pr [X > n] \leq \Pr [X > n]\\
    & \Leftrightarrow & \Pr [X \geq k] \leq 1\,,
  \end{eqnarray*}
  and this last inequality clearly holds.
\end{proof}

We also record the facts below, which follow respectively from the more general result that first-order stochastic dominance is preserved by non-decreasing mappings, and from a coupling argument.
\begin{fact}
\label{fact:min_preserves_dominance_order}
Consider two real-valued random variables $X,Y$, 
and $n \geq 0$. If $X \succeq Y$, then $\min(X,n) \succeq \min(Y,n)$: for all $k$,
\[ \Pr [\min (X, n) \geq k] \geq \Pr [\min (Y, n) \geq k]\,; \]
i.e., the $\min$ operator preserves first-order stochastic dominance relation.
\end{fact}
\begin{fact}
\label{fact:dominance_between_binomials}
Let $X \sim \tmop{Bin} (n, p)$ and $Y \sim \tmop{Bin} (m, p)$, where $m \geq n$. Then $X \preceq Y$.
\end{fact}
\begin{lemma}
  \label{lemma:coupling-for-decoupling-attempt-2}
  Let
  $X_1,\dots,X_m$ be sampled from the sampling process in Definition~\ref{definition:coupling-decoupling-sampling-procedure-attempt-2}, and $I,
  I^c$ be any partition of $[m]$. 
  Define $Z_I \eqdef \sum_{i \in I} X_i$, $Z_{I^c} \eqdef \sum_{i \in I^c} X_i$, and $Y_I
  \sim \min (\tmop{Bin} (| I |, p), n)$, $Y_{I^c} \sim \min (\tmop{Bin} (| I^c
  |, p), n)$. Then 
  \[
    Z_I \cdot Z_{I^c} \preceq Y_I \cdot Y_{I^c}\,.
  \]
\end{lemma}

\begin{proof}
  We prove the lemma by defining a coupling $Z_I, Z_{I^c}, Y_I, Y_{I^c}$ such that $Z_I \cdot Z_{I^c} \leq Y_I \cdot Y_{I^c}$ with probability one. The sampling process below will
    generate samples $(X_i)_{1\leq i\leq m},  Z_I, Z_{I^c}, Y_I, Y_{I^c}$ for all possible
    realizations of $I$ and $I^c$. In other words, from a given sequence $\{
    X_i' \}_{i \in [m]}$, we will generate $\{ X_i \}_{i \in [m]}, Y_{I_1},
    Y_{I^c_1}, Y_{I_2}, Y_{I_2^c}, \ldots, Y_{I_{2^m}}, Y_{I_{2^m}^c}, Z_{I_1},
    Z_{I^c_1}, Z_{I_2}, Z_{I_2^c}, \ldots, Z_{I_{2^m}}, Z_{I_{2^m}^c}$, where the $(I_i, I^c_i)$ enumerate all partitions of $[m]$ in two sets.
    \begin{enumerate}
      \item Initialize $X_i = 0$ for all $i \in [m]$; sample $( X_i' )_{1\leq i\leq m}$ as
      $m$ i.i.d.\ $\tmop{Bern} (p)$;
      
      \item If $\sum_{i \in [m]} X_i' < n$, let $X_i = X_i'$ for all $i \in
      [m]$;
      
      \item If $\sum_{i \in [m]} X_i' \geq n$, let $\mathcal{S}' = \{ i
      \in [m] \of X_i' = 1 \}$ and let $\mathcal{S}$ be a uniformly random subset
      of $\mathcal{S}'$ with size $n$; set $X_i =
      X_i'$ for $i \in \mathcal{S}$.
      
      \item \label{enumerate:the-coupling-step}For each $I \in \{ I_1, \ldots,
      I_{2^m} \}$, denote $\mathcal{S}_I' =\mathcal{S}' \cap I$. Select a
      uniformly random subset of $\mathcal{S}_I'$ with at most $n$ indices which is a
      superset of $\mathcal{S} \cap I$. In more detail, if $| \mathcal{S} \cap
      I | < n$, select $\min (| \mathcal{S}_I' |, n) - | \mathcal{S} \cap I |$
      elements uniformly at random from $\mathcal{S}_I' \setminus (\mathcal{S} \cap
      I)$ to add to $\mathcal{S} \cap I$, which becomes $\mathcal{S}_I$;
      else, let $\mathcal{S}_I =\mathcal{S} \cap I$. Repeat a similar process for $I^c$ to obtain
      $\mathcal{S}_{I^c}$.
      
      \item For each $I \in \{ I_1, \ldots, I_{2^m} \}$, set $Y_I = \sum_{i \in
      \mathcal{S}_I} X_i'$ and $Y_{I^c} = \sum_{i \in S_{I^c}} X_i'$.
    \end{enumerate}
 
  From the above definition, we can
  readily see that for any $I$, $Y_I \geq Z_I$ and $Y_{I^c} \geq
  Z_{I^c}$. What is left is to argue that the $Y_I \sim \min (\tmop{Bin} (| I
  |, p), n)$ and $Y_{I^c} \sim \min (\tmop{Bin} (| I^c |, p), n)$. We start by
  noting that for any $k < n$, $\{ Y_I = k \} = \{ | \mathcal{S}_I | = k \} =
  \{ | \mathcal{S}_I' | = k \}$. The last equality comes from the fact that $|
  \mathcal{S}_I | < n$ can only mean that $| \mathcal{S}_I' | < n$, and the
  selection process in step \ref{enumerate:the-coupling-step} will thus add all elements from
  $\mathcal{S}_I'$ to $\mathcal{S}_I$. From here, we have $\Pr [Y_I = k] = \Pr
  [| \mathcal{S}_I' | = k] = \Pr [\tmop{Bin} (m, p) = k]$, for $k < n$; and we
  have $\Pr [Y_I = n] = 1 - \Pr [Y_I < n] = \Pr [\tmop{Bin} (m, p) \geq
  n]$. As a result, $Y_I \sim \min (\tmop{Bin} (| I |, p), n)$. Similarly, we
  can argue that $Y_{I^c} \sim \min (\tmop{Bin} (| I^c |, p), n)$.
\end{proof}

%% file: app-appendix_deferred_proofs.tex
\subsection{Proof of Lemma~\ref{corollary:hellinger_projection_lb_product}}
\label{proof:lemma:hellinger_projection_lb_bayes_net}
\new{
\CorollaryHellingerProjectionLbProduct*
  \begin{proof}
    Since squared Hellinger distance is an $f$-divergence, by the data processing inequality, we have that
    \begin{equation}\label{eq:hellinger:projection}
    d_{\rm{}H} (P, Q) \geqslant d_{\rm{}H} (\pi P, \pi Q).
    \end{equation}
    By the subadditivity of Hellinger
    \new{
      for Bayes nets from Corollary~\ref{corollary:squared_hellinger_subadditivity} along with~\eqref{eq:hellinger:projection}, we obtain the following:
      \begin{eqnarray*}
        d_{\rm{}H} (P, P') & \leqslant & d_{\rm{}H} (P, Q) + d_{\rm{}H} (Q, P')\\
        & \leqslant & d_{\rm{}H} (P, Q) + \Big(\sum_{i = 1}^n d_{\rm{}H}^2 (P'_{X_i}, Q_{X_i}) \Big)^{1/2}\\
        & = & d_{\rm{}H} (P, Q) + \Big(\sum_{i = 1}^n d_{\rm{}H}^2 (P_{X_i}, Q_{X_i}) \Big)^{1/2}\\
        & \leqslant & d_{\rm{}H} (P, Q) + \sqrt{nd_{\rm{}H}^2 (P, Q)}\\
        & = & \left( 1 + \sqrt{n} \right) d_{\rm{}H} (P, Q).
      \end{eqnarray*}
    }
  \end{proof}
}
\subsection{Proof of Lemma~\ref{lemma:technical_distance_lb_2xbinomials_conditional_tree}}
\new{\DistanceLbBinomialsConditionalTree*}
\begin{proof}
  \label{proof:lemma:technical_distance_lb_2xbinomials_conditional_tree}
  By concentration of Binomials,
  \begin{eqnarray*}
    &  & \left( \frac{1}{2} \right)^n (1 - \delta)^n \sum_{k_1 =
    0}^{a} \sum_{k_2 = 0}^{b} \binom{a}{k_1}
    \binom{b}{k_2} \left| \left( \left( \frac{1 + \delta}{1 - \delta}
    \right)^{k_1 + k_2} - \left( \frac{1 + \delta}{1 - \delta} \right)^{k_1
    + b - k_2} \right) \right|\\
    & \geqslant & \left( \frac{1}{2} \right)^n (1 - \delta)^n \sum_{k_1 =
    \frac{a}{2} + \sqrt{a}}^{\frac{a}{2} + 2 \sqrt{a}}
    \sum_{k_2 = \frac{b}{2} + \sqrt{b}}^{\frac{b}{2} + 2
    \sqrt{b}} \binom{a}{k_1} \binom{b}{k_2} \left| \left(
    \left( \frac{1 + \delta}{1 - \delta} \right)^{k_1 + k_2} - \left(
    \frac{1 + \delta}{1 - \delta} \right)^{k_1 + b - k_2} \right)
    \right|\\
    & \geqslant & \left( \frac{1}{2} \right)^n \frac{C}{\sqrt{a
    b}} \cdot 2^{| a | + | b |} \sum_{k_1 = \frac{a}{2} +
    \sqrt{a}}^{\frac{a}{2} + 2 \sqrt{a}} \sum_{k_2 =
    \frac{b}{2} + \sqrt{b}}^{\frac{b}{2} + 2 \sqrt{b}} (1 -
    \delta)^n \left| \left( \left( \frac{1 + \delta}{1 - \delta}
    \right)^{k_1 + k_2} - \left( \frac{1 + \delta}{1 - \delta} \right)^{k_1
    + b - k_2} \right) \right|\\
    & \geqslant & \frac{C}{\sqrt{a b}} \sum_{k_1 =
    \frac{a}{2} + \sqrt{a}}^{\frac{a}{2} + 2 \sqrt{a}}
    \sum_{k_2 = \frac{b}{2} + \sqrt{b}}^{\frac{b}{2} + 2
    \sqrt{b}} (1 - \delta)^n \left| \left( \left( \frac{1 + \delta}{1 -
    \delta} \right)^{k_1 + k_2} - \left( \frac{1 + \delta}{1 - \delta}
    \right)^{k_1 + k_2 + b - 2 k_2} \right) \right|
  \end{eqnarray*}
  where $C$ is some constant larger than 0. For $k_1 + k_2 - \frac{n}{2} = l \in
  \left[ \sqrt{a} + \sqrt{b}, 2 \left( \sqrt{a} + \sqrt{b}
  \right) \right]$, $l \geqslant \sqrt{a + b} = \sqrt{n}$.
  \begin{eqnarray}
    &  & (1 - \delta)^n \left| \left( \left( \frac{1 + \delta}{1 - \delta}
    \right)^{k_1 + k_2} - \left( \frac{1 + \delta}{1 - \delta} \right)^{k_1 +
    k_2 + b - 2 k_2} \right) \right| \nonumber\\
    & = & \left| \left( (1 - \delta)^n \left( \frac{1 + \delta}{1 - \delta}
    \right)^{n / 2 + l} - (1 - \delta)^n \left( \frac{1 + \delta}{1 - \delta}
    \right)^{n / 2 + l + b - 2 k_2} \right) \right| \nonumber\\
    & = & (1 - \delta^2)^{n / 2} \left( \frac{1 + \delta}{1 - \delta} \right)^l
    \left| 1 - \left( \frac{1 + \delta}{1 - \delta} \right)^{b - 2 k_2}
    \right| \nonumber\\
    & \geqslant & e^{- \delta^2 n} e^{2 \delta \cdot l}  \left| 1 - \left(
    \frac{1 + \delta}{1 - \delta} \right)^{b - 2 k_2} \right| \geqslant e^{2
    \varepsilon - \varepsilon^2}  \left| 1 - \left( \frac{1 + \delta}{1 -
    \delta} \right)^{b - 2 k_2} \right|  \label{eq:bin_asym_1}\\
    & \geqslant & e^{\varepsilon} \left( 1 - \left( \frac{1 + \delta}{1 -
    \delta} \right)^{b - 2 k_2} \right) \geqslant e^{\varepsilon} (1 - e^{4
    \delta (b - 2 k_2)}) .  \label{eq:bin_asym_2}
  \end{eqnarray}
  where (\ref{eq:bin_asym_1}) and (\ref{eq:bin_asym_2}) follows from $1 - x
  \geqslant e^{- 2 x}$, $0 < x < 0.79$ and $e^{4 x} \geqslant \left( \frac{1 +
  x}{1 - x} \right) \geqslant e^{2 x}$, $0 < x < 0.95$. All these inequalities
  hold for $n$ larger some constant and every $\varepsilon \in (0, 1]$. Since
  $b \geqslant \frac{1}{4} n$, and by the summation above $b - 2 k_2 \in
  \left[ - 4 \sqrt{b}, - 2 \sqrt{b} \right]$,
  \[ e^{\varepsilon} (1 - e^{4 \delta (b - 2 k_2)}) \geqslant
     e^{\varepsilon} \left( 1 - e^{- 8 \frac{\varepsilon}{\sqrt{n}} 
     \sqrt{b}} \right) \geqslant e^{\varepsilon} (1 - e^{- 4 \varepsilon})
     \geqslant \varepsilon . \]
  and therefore, summing up every term, we have our lower bound
  \[
  \frac{C}{\sqrt{a b}} \sum_{k_1 = \frac{a}{2} +
    \sqrt{a}}^{\frac{a}{2} + 2 \sqrt{a}} \sum_{k_2 =
    \frac{b}{2} + \sqrt{b}}^{\frac{b}{2} + 2 \sqrt{b}} (1 -
    \delta)^n \left| \left( \left( \frac{1 + \delta}{1 - \delta}
    \right)^{k_1 + k_2} - \left( \frac{1 + \delta}{1 - \delta} \right)^{k_1
    + k_2 + b - 2 k_2} \right) \right| \geqslant C \varepsilon
  \]
  concluding the proof.
\end{proof}

\subsection{Proof of~\eqref{fact:intermedia_upper_bound_on_multigraph_cycle_term:restated}}
\label{app:fact:intermedia_upper_bound_on_multigraph_cycle_term:proof}

\begin{fact}
  \label{fact:intermedia_upper_bound_on_multigraph_cycle_term}
  For any set of cycles such that $\sum_i
  | \sigma_i | \leqslant n$, we have
  \[
    \prod_{\substack{\sigma_i : \tmop{even}\\ | \sigma_i | \geqslant 4}} (1 + (- 4
    \delta)^{| \sigma_i |}) \prod_{\substack{\sigma_i : \tmop{odd}\\ | \sigma_i | \geqslant 4}} (1 - (- 4 \delta)^{| \sigma_i |})
    \leqslant e^{O \left( \varepsilon^5 / n^{\frac{3}{2}} \right)}
    \prod_{\substack{\sigma_i : \tmop{even}\\ | \sigma_i | = 4}} (1 + (4 \delta)^4)
    \prod_{\substack{\sigma_i : \tmop{odd}\\ | \sigma_i | = 4}} (1 - (4 \delta)^4)
  \]
  \begin{proof}
  We have
    \begin{eqnarray*}
      \prod_{\sigma_i : \tmop{even}} &  & (1 + (- 4 \delta)^{| \sigma_i |})
      \prod_{\sigma_i : \tmop{odd}} (1 - (- 4 \delta)^{| \sigma_i |})\\
      & = & \prod_{\substack{\sigma_i : \tmop{even}\\ | \sigma_i | \geq 5}} (1 + (-
      4 \delta)^{| \sigma_i |}) \prod_{\substack{\sigma_i : \tmop{odd}\\ | \sigma_i | \geq 5}} (1 - (- 4 \delta)^{| \sigma_i |})
      \prod_{\substack{\sigma_i : \tmop{even}\\ | \sigma_i | = 4}} (1 + (- 4
      \delta)^{| \sigma_i |}) \prod_{\substack{\sigma_i : \tmop{odd}\\ | \sigma_i | = 4}}
      (1 - (- 4 \delta)^{| \sigma_i |})\\
      & \leqslant & \prod_{\sigma_i : | \sigma_i | > 5} (1 + (4 \delta)^{|
      \sigma_i |}) \prod_{\sigma_i : \tmop{even}, | \sigma_i | = 4} (1 + (- 4
      \delta)^{| \sigma_i |}) \prod_{\sigma_i : \tmop{odd}, | \sigma_i | = 4}
      (1 - (- 4 \delta)^{| \sigma_i |})\\
      & \leqslant & (1 + (4 \delta)^5)^{\frac{n}{4}} \prod_{\sigma_i :
      \tmop{even}, | \sigma_i | = 4} (1 + (- 4 \delta)^{| \sigma_i |})
      \prod_{\sigma_i : \tmop{odd}, | \sigma_i | = 4} (1 - (- 4 \delta)^{|
      \sigma_i |})\\
      & \leqslant & e^{(4 \delta)^5 n / 4} \prod_{\sigma_i : \tmop{even}, |
      \sigma_i | = 4} (1 + (- 4 \delta)^{| \sigma_i |}) \prod_{\sigma_i :
      \tmop{odd}, | \sigma_i | = 4} (1 - (- 4 \delta)^{| \sigma_i |})\\
      & \leq & e^{256\varepsilon^5 / n^{\frac{3}{2}}}
      \prod_{\sigma_i : \tmop{even}, | \sigma_i | = 4} (1 + (- 4 \delta)^4)
      \prod_{\sigma_i : \tmop{odd}, | \sigma_i | = 4} (1 - (- 4 \delta)^4)
    \end{eqnarray*}
    
  \end{proof}
\end{fact}

\section{Structured Testing Lower Bound}\label{sec:gen}

Letting $D = 2^n$, we will rely on the construction from the ``standard'' lower bound of~{\cite{paninski2008coincidence}} by picking a uniformly
random subset $S$ of $\{0, 1\}^n$ of size $\frac{D}{2}$. Denote
$\mathcal{S}$ the set of all such combinations of $S$, and define
$\mathcal{P}_{\tmop{no}}$ to be $\mathcal{P}_{\tmop{no}} \assign \left\{ P =
\frac{1 + C \eps}{2} U_S + \frac{1 - C \eps}{2} U_{S^c} \mid S \in \mathcal{S}
\right\}$, where $C>0$ is a suitable normalizing constant. As before, $U_S$ denotes the uniform distribution on the set of variable $S$
and $P \in \mathcal{P}_{\tmop{no}}$ is a mixture of two uniform
distributions on disjoint parts, with different weights.

It is known that $\Omega (2^{n / 2} / \eps^2)$
samples are required to distinguish between such a randomly chosen $P$ and
the uniform distribution $U$; further, assume we know that the uniform distribution $U$
is in $\mathcal{C}$. What remains to show is the \emph{distance}, that is, ``most'' choices of $P \in
\mathcal{P}_{\tmop{no}}$ are $\eps$-far from $\mathcal{C}$. To argue that
last part, we will use our assumption that $\mathcal{C}$ can be learned with
$m$ samples to conclude by a counting argument: i.e., we will show that there can be at most $2^{mn}$
or so ``relevant'' elements of $\mathcal{C}$, while there are at least $2^{2^{\Omega (n)}}$ $\mathcal{P}_{\tmop{no}}$ that are $\varepsilon$-far from each
other. Suitably combining the two will establish the theorem below:
\begin{theorem}
  \label{theorem:testing_hard_for_easy_to_learn_dist}Let $\mathcal{C}$ be a
  class of probability distributions over $\{0, 1\}^n$ such that the following
  holds: (1) the uniform distribution belongs to $\mathcal{C}$ (2) there
  exists a learning algorithm for $\mathcal{C}$ with sample complexity $m = m
  (n, \eps)$. Then, as long as $mn \ll 2^{O (n)}$, testing whether an
  arbitrary distribution over $\{0, 1\}^n$ belongs to $\mathcal{C}$ or is $\eps$-far
  from every distribution in $\mathcal{C}$ in total variation distance requires
  $\Omega (2^{n / 2} / \eps^2)$ samples.
\end{theorem}

\begin{proof}
  As discussed above, indistinguishability follows from the literature~{\citep{paninski2008coincidence}}, and thus all we need to show now is that $\mathcal{P}_{\tmop{no}}$ is far from every distribution in $\mathcal{C}$.
  By assumption (2), there exists an algorithm $H\colon \{0, 1\}^{mn} \to
  \mathcal{P}$ (without loss of generality, we assume $H$ deterministic) that can output an estimated
  distribution given $m = m (n, \varepsilon)$ samples from $P \in \mathcal{C}$. Thus, for every $P
  \in \mathcal{C}$ given $m$ $i.i.d$. samples $X \in \{0, 1\}^{mn}$, $\Pr_{X \sim P^{\otimes m}} (d_{\tmop{TV}} (H (X), P) < \varepsilon)
  \geqslant 2 / 3$. 
  
  In particular, this implies the weaker statement that, for every $P \in \mathcal{C}$, there
  exists \emph{some} $x$ in $\{0, 1\}^{mn}$ s.t. $P \in B(H (x), \varepsilon)$ (where $B(x,r)$ denotes the TV distance ball of radius $r$ centered at $x$).
  By enumerating all possible values in $\{0, 1\}^{mn}$, we then can obtain an
  $\varepsilon$-cover $\{H (x_1), \ldots, H (x_{2^{mn}})\}$ of $\mathcal{C}$, that is, such that
  $\mathcal{C} \subseteq \bigcup_{i = 1}^{2^{mn}} B (H (x_i), \varepsilon)$. The $\varepsilon$-covering number of $\mathcal{C}$ is thus upper bounded by $2^{O (mn)}$.
  
  Next, we lower bound the size of $\mathcal{P}_{\rm{}no}$ by constructing an $\varepsilon$-packing $P_{\eps}$, where $P_{\eps} = \{P_i \in \mathcal{P}_{\tmop{no}}, i
  \in \mathbb{N}: d_{\tmop{TV}} (P_i, P_j) > \eps, i \neq j\}$.
  For $P,Q \in \mathcal{P}_{\tmop{no}}$ corresponding to two sets $S,S'$, each of size $\frac{D}{2}=2^{n-1}$, we have
  \begin{eqnarray*}
  d_{\tmop{TV}} (P,Q) & = & \frac{1}{2} | S \triangle S' | \cdot \left|
  \frac{1 + C \varepsilon}{2} - \frac{1 - C \varepsilon}{2} \right| \cdot 
  \frac{2}{D} = C\eps\cdot \frac{| S \triangle S' |}{D} > \varepsilon
  \end{eqnarray*}
  For this to be at least $\eps$, the pairrwise
  symmetric difference of (the sets corresponding to the) distributions in $P_{\eps}$ should be at least $\frac{D}{C} = \Omega(2^n)$. We know, by, e.g., \citet[Proposition~3.3]{Blais_2019} that there exist such families of balanced subsets of $\{0,1\}^n$ of cardinality at least $\Omega(2^{2^{\rho n}})$, where 
  $\rho>0$ is a constant that only depends on $C$.
  
  Thus, the size of $\mathcal{P}_{\tmop{no}}$ is itself $\Omega(2^{2^{\rho n}})$; combining this lower bound with the upper bound on the covering number of $\mathcal{C}$ concludes the proof.
\end{proof}

As a corollary, instantiating the above to the class $\mathcal{C}$ of degree-$d$ Bayes nets over $n$ nodes readily yields the following:
\begin{corollary}
  \label{corollary:lb_general}
  For large enough $n$, testing whether an arbitrary distribution over
  $\{0, 1\}^n$ is a degree-$d$ Bayes net or is $\eps$-far from every such
  Bayes net requires $\Omega (2^{n / 2} / \eps^2)$ samples, for any $d = o
  (n)$ and $\eps \geq 2^{- O (n)}$.
\end{corollary}  
  \begin{proof}
    We can obtain a learning upper bound of $m = O (2^d n \log (2^{d + 1} n)
    \log (n^{d n}) / \varepsilon^2)$ for degree-$d$ Bayes nets by combining
    the known-structure case (proven in~{\citet{BhattacharyyaGMV20}}) with the
    reduction from known-structure to unknown-structure (via hypothesis
    selection/tournament~{\citep{CanonneDKS20}}). We have $m n \leqslant O
    (2^d n^6 / \varepsilon^2)$. To have $2^{mn} \ll 2^{2^{\rho n}}$, where
    $\rho$ is some constant, we need $m n < 2^{O (n)}$, which requires $d = o
    (n)$ and $\varepsilon \geq 2^{- O (n)}$ for large enough $n$.
  \end{proof}

%% file: sec-lower_bound_warmup.tex
In this section, we state and prove a simpler, but quantitatively weaker lower bound than Theorem~\ref{theo:main:lb} for independence testing, Theorem~\ref{theorem:lb_2}. This simpler lower bound is adapted from~\citet[Theorem~13]{CanonneDKS20} -- the ``mixture-of-products'' construction. Their analysis readily provides indistinguishability, and distance \emph{from the uniform distribution}. Thus, all we need here is to show that most of these hard instances (i.e., ``mixtures of products'') are far from every product distribution (Lemma~\ref{lemma:technical:mixture_of_products_far_from_product_distributions}), not just the uniform distribution. While the $\Omega ( 2^{d/2} \sqrt{n} / \eps^2 )$ lower bound this yields is not as tight in terms of sample complexity, with a $\sqrt{n}$ dependence instead of $n$ (at a high level, this is because we fix the Bayesian structure, and thus the algorithms have additional information they can leverage), the restriction on $d$ is much milder than the one in Theorem~\ref{theo:main:lb}, allowing up to $d = n/2$.

\begin{theorem}
    \label{theorem:lb_2}
    Let $1\leq d \leq n / 2$. Testing whether an arbitrary
    degree-$d$ Bayes net over $\{0, 1\}^n$ is a product distribution or is
    $\eps$-far from every product distribution requires $\Omega (2^{d / 2} 
    \sqrt{n} / \eps^2)$ samples. This holds even if the structure of the
    degree-$d$ Bayes net is known.
\end{theorem}
\begin{proof}
As discussed above, we will use the same ``mixture-of-products'' construction as in~\citet[Theorem~13]{CanonneDKS20}, which established a lower bound of $\Omega (2^{d / 2} \sqrt{n} / \eps^2)$ samples to distinguish it from the uniform distribution. We first recall the definition of this ``mixture-of-products'' construction.

Letting $N \eqdef n-d\geq n/2$, we define, for $z\in\{\pm 1\}^N$ the product distribution $p_Z$ over $\{0,1\}^N$ by
\begin{equation}
    \label{eq:far:products}
    p_z(x) = \prod_{i = 1}^N \mleft(\frac{1}{2} + z_i (- 1)^{x_i} \delta\mright), \qquad x\in\{0,1\}^N.
\end{equation}
where $\delta \eqdef \frac{\eps}{\sqrt{N}} = \Theta\mleft(\frac{\eps}{\sqrt{n}}\mright)$. A mixture-of-products distribution is then defined by choosing $2^d$ i.i.d.\ $Z_1,\dots, Z_{2^d}\in\{\pm 1\}^N$ uniformly at random, and setting $\p_{Z_1,\dots, Z_{2^d}}$ to be the distribution over $\{0,1\}^n$ which is uniform on the first $d$ bits, and where the first $d$ bits of $x$ are seen as the binary representation (i.e., a ``pointer'') for which $\p_{Z_i}$ will be used for the last $N$ bits of $x$. That is,
\begin{equation}
    \label{eq:mixture:products}
        \p_{Z_1,\dots, Z_{2^d}}(x) = \frac{1}{2^d}p_{Z_{\iota(x_1,\dots, x_d)}}(x_{d+1},\dots,x_{n}), \qquad x\in\{0,1\}^n
\end{equation}

where, analogously to Definition~\ref{def:mixture:trees}, $\iota\colon \{0,1\}^d \to [2^d]$ is the indexing function, mapping the binary representation (here on $d$ bits) to the corresponding number.

As mentioned in the preceding discussion, this construction was already used in~\citet[Theorem~13]{CanonneDKS20}, where the authors show an $\Omega ( 2^{d/2} \sqrt{n} / \eps^2 )$  sample complexity lower bound to distinguish a uniformly randomly chosen mixture-of-products distribution (which is a degree-$d$ Bayes net) from the uniform distribution (which is a product distribution). For their theorem (a lower bound on testing \emph{uniformity}), they then conclude from the easy fact that every such mixture-of-products distribution is $\eps$-far from the uniform distribution. This is not enough for us, as, to obtain the lower bound stated in Theorem~\ref{theorem:lb_2}, what we need is to show that every such mixture-of-products distribution (or at least \emph{most} of them) is far from \emph{every} product distribution, not just the uniform one. This is the only missing part towards proving Theorem~\ref{theorem:lb_2}, and is established in our next lemma:
  \begin{lemma}[Distance from Product distributions]
    \label{lemma:technical:mixture_of_products_far_from_product_distributions}
    For $p$ uniformly sampled from the mixture-of-products construction,
    \[ 
        \Pr \left[ \min_{q_1,q_2}d_{\tmop{TV}} (p, q_1 \otimes q_2) \geq \frac{\varepsilon}{750} \right] \geq \frac{9}{10}
    \]
    as long as $n \geq d + C_1$, for some constants $C_1 > 0$ and $n / 2 \geq d$.
\end{lemma}
This lemma will directly follow from Claim~\ref{claim:technical:distance_of_conditionals} (below) and Lemma~\ref{lemma:technical_approximated_lb_product_distributions}; the rest of this appendix is thus dedicated to proving the former, which states that most mixture-of-products distributions are far from the product of their marginals.

\begin{claim}
\label{claim:technical:distance_of_conditionals} 
Given a mixture-of-products distribution $p$ as in~\eqref{eq:mixture:products}, let $p_1$ be the marginal of $\p$ on the first $d$ variables (parent nodes) and $p_2$ the marginal on the $N$ last variables (child nodes). Note that $p_1\otimes p_2$ is then a product distribution on $\{0,1\}^n$.  Then, we have
\[ \Pr \left[ d_{\tmop{TV}} (p, p_1 \otimes p_2) \geq \frac{\varepsilon}{250}
    \right] \geq \frac{9}{10} \]
as long as $n \geq C_1 + d$, for some constants $C_1 \geq 0$ and $n \geq 2d$.
\end{claim}
\begin{proof}
  Fix any mixture-of-products distribution $p$. From Lemma~\ref{lemma:technical:TV_lower_bound_of_its_marginals} and the structure of $p$ as given in~\eqref{eq:mixture:products}, one can show that
  \[
  d_{\tmop{TV}} (p, p_1 \otimes p_2) \geq \frac{1}{2^{d - 1}}  \sum_{x_2,
    \ldots, x_d} d_{\tmop{TV}} (p (\cdummy \mid 0, x_2, \ldots, x_d), p
    (\cdummy \mid 1, x_2, \ldots, x_d)).
\]
Denoting $p (\cdummy \mid 0, x_2, \ldots, x_d)$ by $p_{\iota(x_2, \ldots, x_d)}$  and $p (\cdummy \mid 0, x_2, \ldots, x_d)$ by $q_{\iota(x_2, \ldots, x_d)}$ (where $\iota(x_2, \ldots, x_d)\in[2^{d-1}]$, abusing slightly the definition of the indexing function to extend it to $d-1$ bits), we can rewrite this as
  \[ d_{\tmop{TV}} (p, p_1 \otimes p_2) \geq \frac{1}{2^{d - 1}}  \sum_{t =
    1}^{2^{d - 1}} d_{\tmop{TV}} (p_t, q_t) \eqqcolon \overline{d_{\tmop{TV}}}\,.
  \]
  Now, since $d \leqslant n/2$, one can show that, for every fixed $t$, 
  \begin{equation}\label{lemma:single_pz_farness}
  \Pr [d_{\tmop{TV}} (p_t, q_t) < \varepsilon / 25] < e^{- C\cdot N}
  \end{equation}
  where the probability is taken over the choice of $p$ (i.e., its $2^d$ parameters $Z_1,\dots, Z_{2^d}$), and $C>0$ is an absolute constant.  We defer the proof of this inequality to the end of the appendix, and for now observe that the RHS is less than $1 / 10$ for $N$ greater than some (related) absolute constant $C_1>0$. We can then write, letting $D \eqdef 2^{d - 1}$ and $X_t \eqdef
  \mathbbm{1} \{d_{\tmop{TV}} (p_t, q_t)\}$
  \[ \overline{d_{\tmop{TV}}} \geq \frac{\varepsilon}{25} \cdot \frac{1}{2^{d
    - 1}}  \sum_{t = 1}^{2^{d - 1}} \mathbbm{1} \left\{ d_{\tmop{TV}}
    (p_t, q_t) \geq \frac{\varepsilon}{25} \right\} = \frac{\varepsilon}{25}
    \cdot \frac{1}{D}  \sum_{t = 1}^D X_t, \]
  where the $X_t$'s are i.i.d. Bernoullis with, by the above analysis, parameter $\alpha \geq 1 - e^{-
  C \cdot N} \geq 9 / 10$. We then have
    \[ 
  \Pr [d_{\tmop{TV}} (p, p_1 \otimes p_2) < \varepsilon / 250] \leq \Pr
    [\overline{d_{\tmop{TV}}} < \varepsilon / 250] \leq \Pr \left[ \sum_{t =
    1}^D X_t < \frac{D}{10} \right] ,
    \]
    so it remains to show that the RHS is less than $1/10$. Since $\mathbb{E}[X_t] \geq 9/10$ for all $t$, this readily follows from a Hoeffding bound, for $d\geq 1$.
\end{proof}
To conclude, we only need to prove~\eqref{lemma:single_pz_farness}, which (slightly rephrasing it) tells us that two independent parameterizations $p_Z, p_Z'$ will be at total variation distance at least $\Omega(\eps)$ far with overwhelming probability.
\begin{proof}[Proof of \eqref{lemma:single_pz_farness}]
Let distribution $p_Z, p_{Z'}$ be defined as
  in~\eqref{eq:far:products}, and $Z, Z'$ be i.i.d.\ and uniform on $\{\pm 1\}^N$. The statement to show is then
\begin{equation}
    \label{eq:whattoshow}
  \Pr\mleft[d_{\tmop{TV}} (p_Z, p_{Z'}) \geq \frac{\eps}{25}\mright] \geqslant 1 - e^{-N/18},
\end{equation}
We know (see, e.g., \citet[Lemma 6.4]{kamath2019privately}) that as long as $\delta \leqslant 1 / 6$ (which holds for $n\geq 36$), then the TV distance is related to the $\ell_2$ distance between mean vectors $\mu_Z,\mu_{Z'}\in[0,1]^N$ as
    \begin{equation}
      d_{\tmop{TV}} (p_Z, p_{Z'}) \geq \frac{1}{20}\sqrt{\sum_{i = 1}^N (\mu_{Z,i} - \mu_{Z',i})^2} .
      \label{eq:pz_TV_l2}
    \end{equation}
    Relating this $\ell_2$ distance between mean vectors the Hamming distance between $Z$ and $Z'$, we have
    \begin{equation}
      \sqrt{\sum_{i = 1}^N (\mu_{Z,i} - \mu_{Z',i})^2} = \sqrt{\sum_{i = 1}^N
      \left( \left( \frac{1}{2} - Z_i \delta \right) - \left( \frac{1}{2}-
      Z_i' \delta \right) \right)^2} = 2 \delta \sqrt{\tmop{Hamming} (Z, Z')} ,
      \label{eq:pz_l2_hamming}
    \end{equation}
    where $\tmop{Hamming} (Z, Z') = \sum_{i = 1}^N \mathbbm{1}  [Z_i \neq
    Z_i']$. Noting that Hamming$(Z, Z') \sim \tmop{Bin} (N, 1 / 2)$, we have
    the following via Hoeffding's inequality along with (\ref{eq:pz_TV_l2})
    and (\ref{eq:pz_l2_hamming}),
    \[ 
    \Pr \left[ d_{\tmop{TV}} (p_Z, p_{Z'}) \geqslant \frac{1}{20} 
       \sqrt{\frac{4 \delta^2 N}{3}} \right] \geqslant \Pr [\tmop{Hamming} (Z,
       Z') \geqslant N / 3] \geqslant 1 - e^{- N / 18} . 
       \]
    Since $n / 2 \geqslant d$ and thus, $\sqrt{\frac{4 \delta^2 N}{3}}
    \geqslant \frac{4}{5} \varepsilon$, we get~\eqref{eq:whattoshow}.
  \end{proof}
\noindent This concludes the proof of Lemma~\ref{lemma:technical:mixture_of_products_far_from_product_distributions}, and with it of Theorem~\ref{theorem:lb_2}.
\end{proof}